\documentclass[10pt,journal,cspaper,compsoc]{IEEEtran}
 \usepackage[cmex10]{amsmath}
 \usepackage{amsfonts}
 \usepackage{graphicx}
 \usepackage{algorithm}
\usepackage[algo2e, ruled, vlined]{algorithm2e}
 \usepackage{wrapfig}
 \usepackage{url}
 \usepackage{times}
\usepackage{epsfig}
 \usepackage{ifpdf}
 \ifpdf\fi

\newtheorem{thm}{Theorem}[section]
\newtheorem{theorem}{Theorem}[section]

\newtheorem{lemma}[thm]{Lemma}

\newtheorem{example}[thm]{Example}

\newtheorem{corollary}[thm]{ Corollary}

\newcommand{\ignore}[1]{}
\newcommand{\notinproc}[1]{#1}
\newcommand{\onlyinproc}[1]{}

\newcommand\E{\textsf{E}}
\newcommand{\var}{\mathop{\sf Var}}

\newcommand{\kth}{\text{k}^{\text{th}}}
\newcommand{\ppart}{\text{{\sc bucket}}}
\newcommand{\ADS}{\mathop{\rm ADS}}
\begin{document}

 \title{All-Distances Sketches, Revisited: \\ HIP Estimators for
Massive Graphs Analysis}

 \ignore{
 \numberofauthors{1}
\author{
\alignauthor Edith Cohen\\
       \affaddr{Microsoft Research}\\
       \affaddr{Mountain View, CA, USA}\\
       \email{edith@cohenwang.com}
} }

\author{Edith Cohen\thanks{Tel Aviv University, Israel 
{\tt edith@cohenwang.com}}
} 


\IEEEcompsoctitleabstractindextext{
\begin{abstract}
Graph datasets with billions of edges, such as social and Web graphs,
are prevalent, and scalability is critical.  All-distances sketches
(ADS) [Cohen 1997],  are a powerful tool for scalable
approximation of statistics.

The sketch is a small size sample of the distance relation of a node which emphasizes closer
nodes. Sketches for all nodes are computed using a nearly linear
computation and estimators are applied to sketches of nodes to
estimate their properties.

We provide, for the first time, a unified exposition of ADS algorithms
and applications.  We present the Historic Inverse Probability (HIP)
estimators which are applied to the ADS of a node to estimate a large
natural class of statistics.  For the important special cases of
neighborhood cardinalities (the
number of nodes within some query distance)  and closeness centralities, HIP estimators
have at most half the variance of previous 
estimators and we show that this is essentially optimal.  Moreover,
HIP obtains a polynomial improvement for more general statistics and the
estimators are simple, flexible, unbiased, and elegant.

  For approximate distinct counting on data streams, HIP outperforms
the original estimators for the HyperLogLog MinHash
  sketches (Flajolet et al. 2007), obtaining significantly
improved estimation quality for this state-of-the-art practical algorithm.
\end{abstract}
}

 \maketitle

\section{Introduction}
Massive graph datasets are prevalent and include social and
Web graphs.  Structural properties and relations of nodes in the graph
successfully model 
relations and properties of entities in the  underlying network.
In particular,  many key properties are formulated in terms of 
the set of shortest-path distances $d_{ij}$ (from node $i$ to
$j$) \cite{Bavelas:HumanOrg1948}.


The distance distribution of a node specifies, 
for $d \geq 0$,  the cardinality of the
$d$-neighborhood of the node (number of nodes that are of distance at
most $x$).  We can also consider
the distance distribution of the whole graph, which is the number of pairs of nodes for
each distance $d_{ij}< d$.  
 More generally, we can consider distance-based statistics of the form
\begin{equation} \label{Cquery}
Q_g(i) = \sum_{j| d_{ij}<\infty} g(j,d_{ij}) \  ,
\end{equation}
where  $g(j,d_{ij}) \geq 0$ is a function over both node IDs (meta data)
 and distances.  
Choosing $g(d_{ij})=d_{ij}$, $Q_g(i)$ is the sum of distances from $i$, which is
(the inverse of) the classic closeness centrality measure
\cite{Bavelas:HumanOrg1948,EW_centrality:SODA2001,classiccloseness:COSN2014}.
Centrality which
decays relevance with distance
\cite{Rosenblatt:stats1956,CoSt:pods03f,CoKa:jcss07} and weighs nodes
by their meta-data is expressed using
\begin{equation} \label{closenesscdef}
C_{\alpha,\beta}(i) = \sum_{j| d_{ij}< \infty} \alpha(d_{ij}) \beta(j) \ ,
\end{equation}
where $\alpha \geq 0$ is
monotone non-increasing {\em kernel} and $\beta\geq 0$ is a 
function over node IDs which weighs (or filters)
nodes based on their meta-data.  For example, $\beta$ can 
depend on gender, locality, activity, or age in a social network or a
topic in a Web graph.
When using $\beta\equiv 1$, we can express the
$d$-neighborhood cardinality
using $\alpha(x)=1$ if $x\leq d$ and
$\alpha(x)=0$ otherwise; the set of reachable nodes from $i$ using
$\alpha(x) \equiv 1$; exponential attenuation using
$\alpha(x) = 2^{-x}$ \cite{Dangalchev:2006}; and (inverse) harmonic
mean of distances using $\alpha(x) = 1/x$ \cite{Opsahl:2010,BoldiVigna:IM2014}.

 These distance-based statistics can be computed exactly using
shortest paths computations.  When there are many queries, however, this is computationally
expensive on very large networks.
Efficient algorithms which approximate the distance distributions were
proposed in the last two decades
\cite{ECohen6f,PGF_ANF:KDD2002,bottomk07:ds,CGLM:ICSR2011,hyperANF:www2011}.
Implementations \cite{BackstromBRUV12:websci2012} based on ANF
\cite{PGF_ANF:KDD2002} and hyperANF \cite{hyperANF:www2011} , and more
recently, \cite{CDFGGW:COSN2013}, based on
\cite{ECohen6f,bottomk07:ds}, can process graphs with billions of
edges.

At the core of  all these algorithms
are  {\em All Distances Sketches (ADS)} \cite{ECohen6f}.
The ADS of
a node $v$ contains a random sample of nodes, where the inclusion
probability of a node $u$ decreases with its distance from $v$ (more precisely,
inversely proportional to the number of nodes closer to $v$ than
$u$).  For each included node, the ADS also contains its distance from $v$.
The sketches of different nodes are {\em coordinated}, which means that the
inclusions of each particular node in the sketches of other nodes are positively
correlated.  While coordination \cite{BrEaJo:1972} is an artifact of the way the sketches are computed 
(we could not compute independent sketches as efficiently), it 
enables further applications such as estimating
similarity between neighborhoods of two nodes \cite{ECohen6f}, 
distances,  closeness similarities \cite{CDFGGW:COSN2013}, and
timed influence \cite{DSGZ:nips2013,timedinfluence:2014}.

  An ADS extends
the simpler {\em
  MinHash} sketch \cite{FlajoletMartin85,ECohen6f}
(The term  {\em least element lists} was used in \cite{ECohen6f}, the
term min-wise/min-hash was coined later by Broder \cite{Broder:CPM00}) of
subsets, and inherits its properties.
MinHash sketches are extensively used for approximate distinct
counting \cite{FlajoletMartin85,ECohen6f,DurandF:ESA03,BJKST:random02}  and similarity estimation
\cite{ECohen6f,CoWaSu,Broder:CPM00,BRODER:sequences97} and come in three 
flavors, which correspond to sampling schemes:
A $k$-min sketch \cite{FlajoletMartin85,ECohen6f} is a $k$ sample obtained with replacement, a
bottom-$k$ sketch \cite{ECohen6f} is a $k$ sample without replacement, and a
$k$-partition sketch \cite{FlajoletMartin85} samples one from each of the $k$ buckets in a
uniform at random partition of elements to buckets.  The three flavors provide
different tradeoffs between update costs, information, and maintenance
costs that are suitable for different applications.  In all three, 
the integer parameter $k\geq 1$ 
controls a tradeoff between the information content (and accuracy of
attainable estimates) of
the sketch and resource usage (for storing, updating, or querying
the sketch). 

   The ADS of a
node $v$ can be viewed as representing  MinHash sketches for
all neighborhoods of $v$, that is, it includes a
MinHash sketch of the set of the $i$ closest nodes to $v$, 
for all possible values of $i$.
Accordingly, ADSs come in the same three flavors:
$k$-mins
\cite{ECohen6f,PGF_ANF:KDD2002}, bottom-$k$
\cite{ECohen6f,bottomk07:ds}, and $k$-partition
\cite{hyperANF:www2011}, and have expected size $\leq k\ln n$ \cite{ECohen6f,bottomk:VLDB2008}.  

  We make several contributions.  Our first contribution is a
unified exposition,  provided for the first time,  of ADS flavors
 (Section~\ref{ADSdef:sec}), their relation to MinHash sketches, and
 scalable algorithms (Section~\ref{comptADS:Sec}). This view
 facilitated our unified derivation of general estimators, and we hope
would enable future research and applications of these versatile
 structures.


\ignore{
which compute the 
set of ADSs 
are based on classic shortest-paths algorithms: {\sc PrunedDijkstra} \cite{ECohen6f}, 
which performs pruned applications of Dijkstra's algorithm (or Breadth First Searches  for unweighted graphs).
\cite{ECohen6f}, {\sc DP}
\cite{PGF_ANF:KDD2002,hyperANF:www2011}, which uses dynamic 
programming and applies when edges are unweighted (ADS computation is
implicit
in \cite{PGF_ANF:KDD2002,CGLM:ICSR2011,hyperANF:www2011}, as entries 
are computed but not retained.), and {\sc LocalUpdates}, which
applies dynamic programming to weighted graphs.
With unweighted graphs, these algorithms perform $O(k m \log n)$ edge relaxations (where $m$ is 
 the number of edges and $n$ the number of nodes), and this is also their
main-memory single-processor running times. {\sc LocalUpdates} is
node-centric and appropriate for MapReduce and similar platforms
\cite{pregel:sigmod2010,Naiad:sosp2013}.
}

Our main technical contributions are 
{\em estimators}  for distance-based statistics.
Our estimators are simple, elegant,  and designed for practice:
getting the most from the information present 
in the sketch, in terms of minimizing variance.  
We now provide a detailed overview of
our contributions and the structure of the paper.



\smallskip
\noindent
{\bf MinHash cardinality estimators:}    
Prior to our work, ADS-based neighborhood cardinality estimators 
\cite{ECohen6f,bottomk07:ds,PGF_ANF:KDD2002,hyperloglog:2007,hyperANF:www2011}
were applied by obtaining
from  the ADS a corresponding MinHash 
sketch of the neighborhood and applying a cardinality estimator 
\cite{ECohen6f,FlajoletMartin85,DurandF:ESA03,hyperloglog:2007} to that 
MinHash sketch.  We refer to these estimators as {\em basic}.  

In Section \ref{ADSestNSbasic:sec} we review cardinality estimators,
focusing on variance. Since we are interested in relative error, we 
use the Coefficient of Variation (CV), which is the ratio of the standard
deviation to the mean.  
Our treatment of basic estimators facilitates their comparison to
the new estimators we propose here.
The first-order term (and an upper bound) on the  CV  of the basic
estimators is $1/\sqrt{k-2}$ (see \cite{ECohen6f} for $k$-mins
sketches and extension to bottom-$k$ sketches here).
We show, by applying  the Lehmann-Scheff\'e theorem \cite{LehSche1950},
that these estimators are 
the (unique) optimal unbiased estimators, in terms of minimizing
variance.



\smallskip
\noindent
{\bf Historic Inverse Probability (HIP) estimators:}    
 In Section \ref{ADSRCest:sec} we present the novel HIP estimators.  
HIP estimators 
 improve over the basic cardinality estimators by
utilizing all information present in the ADS (or accumulated during
its computation), rather than just using the
MinHash sketch of the estimated neighborhood.  

 The key idea behind HIP is to consider,
for two nodes $i,j$, the inclusion probability of $j$
in $\ADS(i)$.  Although this probability can not be computed from the
sketch, it turns out that we can still work with it if
we condition it on the randomization of nodes other than $j$.  We refer to
this conditioned inclusion as the {\em HIP probability} and can compute it when $j\in \ADS(i)$.
We can now obtain for each node $j$
a nonnegative estimate $a_{ij} \geq 0$, which we refer to as {\em
  adjusted weight} on its presence with respect to $i$.  The adjusted
weight $a_{ij}$  is equal
to the inverse of the HIP probability when $j\in \ADS(i)$ and $0$
otherwise and is unbiased (has expectation $1$ for any $j$
reachable from $i$). 

 The HIP estimator for the cardinality of the $d$-neighborhood of $i$
 is simply the sum of the adjusted weights of nodes in $\ADS(i)$ that are of
distance at most $d$ from $i$.  We show that 
the HIP estimators obtain a factor-2
reduction in variance over basic estimators, with CV  upper bounded by
the first-order term $1/\sqrt{2(k-1)}$).
We further show that our HIP
 estimators are essentially optimal for ADS-based neighborhood cardinality
estimates, and nearly match an asymptotic (for large enough
cardinality) lower bound of $1/\sqrt{2k}$
on the CV.  Moreover, the HIP estimates can accelerate existing
implementations (ANF \cite{PGF_ANF:KDD2002} and hyperANF
\cite{hyperANF:www2011}) 
essentially without changing the computation.
Our simulations demonstrate
a factor $\sqrt{2}$ gain in both mean square error and mean relative error of HIP over basic estimators.

 Moreover, our HIP estimates also provide us with
 unbiased and nonnegative estimates for
general distance-based statistics such as \eqref{Cquery} and
\eqref{closenesscdef}.
The HIP estimate for \eqref{Cquery}  is
$\hat{Q}_g(i) = \sum_{j\in \ADS(i)}
 a_{ij} g(j,d_{ij})$, which is a sum over (the logarithmically many)
 nodes in $\ADS(i)$.  Similarly, the HIP  estimator for
 \eqref{closenesscdef} is
\begin{equation} \label{centralityest}
\hat{C}_{\alpha,\beta}(i) = \sum_{j\in\ADS(i)} a_{ij} \alpha(d_{ij})
 \beta(j)\ .
\end{equation}

In \cite{CoSt:pods03f,CoKa:jcss07} we estimated  \eqref{closenesscdef}
from $\ADS(i)$ for any (non-increasing) $\alpha$.  The handling of a
general $\beta$, however, required an ADS computation specific to
$\beta$ (see Section \ref{nonuniformweights:sec}).  
The estimators used a reduction to
basic neighborhood cardinality estimators, inheriting the CV of  $1/\sqrt{k-2}$.  
On the same problem, our ADS HIP  estimator \eqref{centralityest}
has CV that is upper bounded by
 $1/\sqrt{2(k-1)}$. Moreover, the HIP estimator applies also when
 the filter $\beta$ in \eqref{closenesscdef} (or the function $g$ in \eqref{Cquery})
are specified after the sketches are computed.  This flexibility of
using the same set of sketches for many queries is important in
the practice of social networks or Web graphs analysis.
For these queries, our HIP estimators obtain up to an $(n/k)$-fold 
improvement in variance over state of the art, which we believe is a subset-weight estimator applied to
the MinHash sketch of all reachable nodes (by taking the average of $g(d_{ij},j)$ over the $k$ samples, multiplied by a cardinality estimate of the number of reachable nodes $n$). 

\ignore{
\smallskip
\noindent
{\bf Permutation estimators:}  
The basic and HIP estimators have CV that is essentially 
independent of cardinality (the neighborhood size). 
The HIP {\em permutation estimator},  presented in Section
\ref{permest:sec}, exploits knowledge of an
upper bound on the domain size (total number of nodes) to obtain more accurate
cardinality estimates than plain HIP for sets that comprise
a good fraction (experimentally, at least 20\%) of the domain.
}

\smallskip
\noindent {\bf ADS for data streams:} Beyond graphs, an ADS can be
viewed as a sketch of an ordered set of elements. HIP is applicable
in all application domains where such ADS sketches can be computed.
Another important application
domain is data streams, where elements populate stream entries.
Instead of distance, we can consider elapsed time from the start to
the first occurrence of the element, which emphasizes earlier entries.
  Alternatively, we can consider
elapsed time from the most recent occurrence to the current time,
which emphasizes recent entries (appropriate for time decaying
statistics \cite{CoSt:pods03f}).
In both cases, an  ADS can be computed
efficiently over the stream \notinproc{(see details in Section \ref{streamADS:sec})}, and we can apply HIP to estimate
statistics of the form \eqref{Cquery} and \eqref{closenesscdef}, where the respective notion
of elapsed time replaces distance in the expressions.  
A stream statistics which received enormous attention, and we treat in detail,
is distinct counting.

\smallskip
\noindent
{\bf HIP estimators for approximate distinct counting:}  
Almost all streaming distinct counters
\cite{FlajoletMartin85,ECohen6f, DurandF:ESA03,BJKST:random02,KNW:PODS2010,hyperloglogpractice:EDBT2013}
maintain a MinHash sketch of the distinct
 elements.
To answer a query (number of distinct elements seen so
 far),  a ``basic'' estimator is applied to the sketch.
In Section \ref{distinctcount:sec} we instead apply
our HIP estimators.  To do that, we consider 
the sequence of elements which invoked an update of the MinHash sketch
over time.  These updates corresponds to entries in the ADS computed with respect to
elapsed time from the start of the stream to the first occurrence of
each element.
Even though the ADS entry is not retained, (the streaming algorithm only
 retains the MinHash sketch), we can compute the
adjusted weight of the new distinct element that  invoked the update.
These adjusted weights are added up to obtain a running estimate on
the number of distinct elments seens so far.
To apply HIP, we therefore need to maintain the MinHash sketch
and an additional approximate (non-distinct) counter, which maintains an
approximate count of distinct elements.  The approximate counter  is updated (by
a positive amount which corresponds to the adjusted weight of the element) each time the sketch is updated.

 We experimentally compare our HIP estimator to the
HyperLogLog approximate distinct counter \cite{hyperloglog:2007},
which is considered to be the state of
the art practical solution.
To facilitate comparison, we apply HIP to the same MinHash sketch
with the same parametrization that the HyperLogLog estimator was designed for.  Nonetheless, we demonstrate significantly
more accurate estimates using HIP.
Moreover, our HIP estimators are unbiased, principled, and do not require
ad-hoc corrections.  They are flexible in that they apply to all
MinHash sketch flavors and can be further
parametrized according to application needs or to obtain even better
accuracy for the same memory.

\smallskip
\noindent
{\bf Approximate counters:}  
In Section \ref{approxcount:sec} we consider the simpler problem of approximate (not
distinct) counters.  Approximate counters are a component of our HIP
approximate distinct counters, and have many other important applications.
We extend the classic counters of Morris \cite{Morris77}  and Flajolet
\cite{Flajolet:BIT85} to apply for arbitrary positive increases, and
addition of approximate counters, using again, inverse probability estimation.

Lastly, we 
present a cardinality estimator which applies only
to the $\ADS$ size (Section \ref{adssizeest:section}) and
discuss extension to non-uniform node weights
(Section \ref{nonuniformweights:sec}).


\section{All-distances sketches}  \label{ADSdef:sec}

 We start with a brief review of MinHash sketches.  The MinHash
 sketch is a randomized summary of a subset $N$ of
items (from some domain $U$) and comes in
three flavors: $k$-mins, $k$-partition, and bottom-$k$, where the
parameter $k$ determines the sketch size.

The sketch is defined with respect to (one or more,
depending on flavor) random permutations of the domain $U$.  It is
convenient to specify a permutation by assigning random {\em rank}
values, $r(j)\sim U[0,1]$, to items.  The permutation is the list of
items sorted by increasing rank order.  To specify multiple
permutations, we use multiple rank assignments.  A $k$-mins sketch
\cite{FlajoletMartin85,ECohen6f} includes the smallest rank in
each of $k$ independent permutations and corresponds to sampling $k$
times with replacement.  A $k$-partition sketch
\cite{FlajoletMartin85,hyperloglog:2007,LOZ:NIPS2012} first maps items uniformly at random to $k$
buckets and then includes the smallest rank in each bucket.
A bottom-$k$ sketch \cite{ECohen6f,BRODER:sequences97} (also known as 
KMV sketch \cite{BJKST:random02}, coordinated order samples
\cite{BrEaJo:1972,Rosen1997a,Ohlsson_SPS:1998}, or CRC \cite{LiChurchHastie:NIPS2008}) 
includes the
$k$ smallest ranks in a single permutation and corresponds
to sampling $k$ times without replacement. 
For 
$k=1$, all three flavors are the same.

MinHash sketches of different subsets $N$ are {\em coordinated}
if they are generated using the same random permutations (or mappings)
of the domain $U$. The concept of sample coordination was introduced by
\cite{BrEaJo:1972}, and in the context of sketches, by \cite{ECohen6f}.

\begin{wrapfigure}{r}{0.2\textwidth}
\centerline{
\ifpdf
\includegraphics[width=0.16\textwidth]{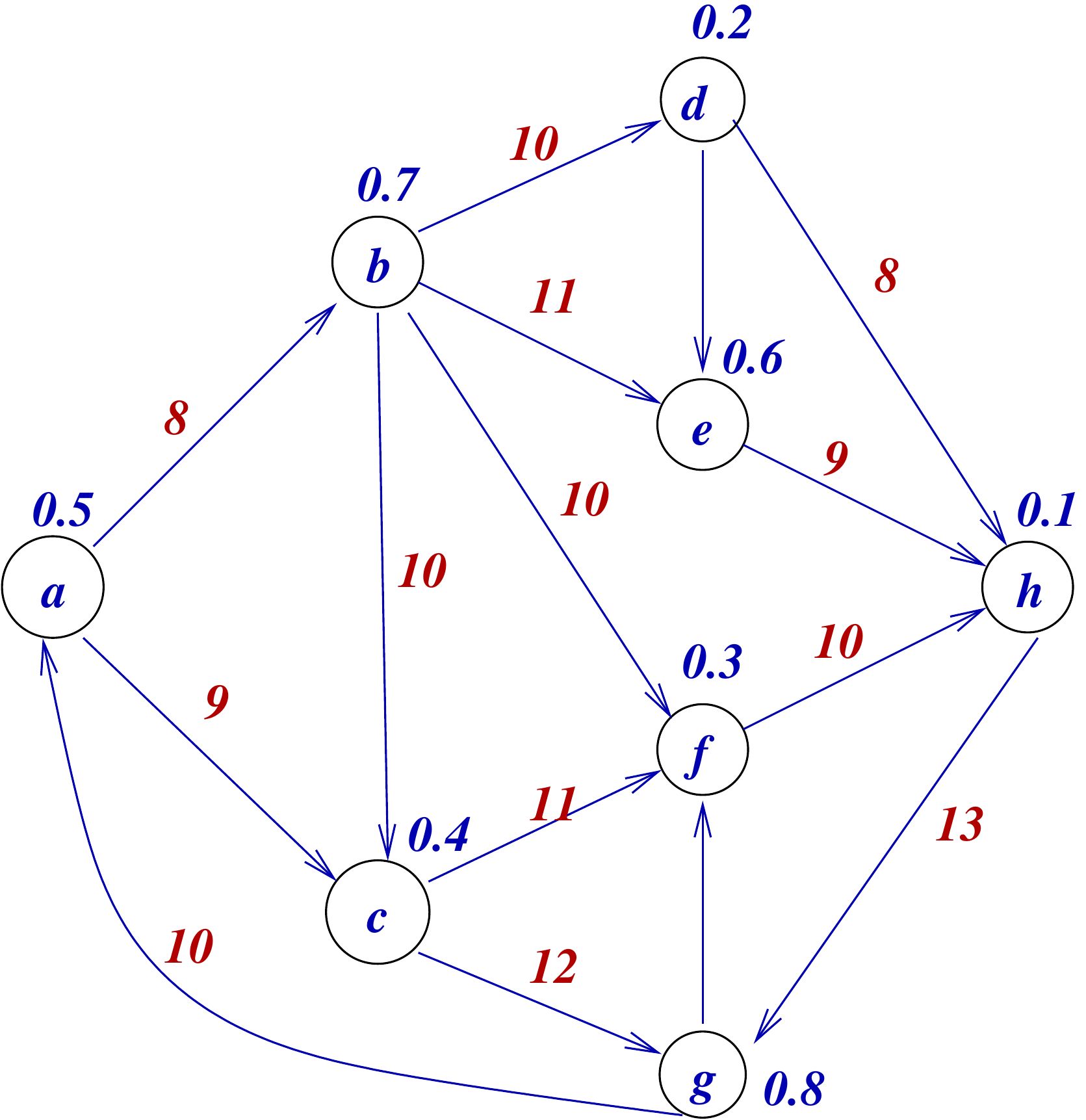} 
\else
 \epsfig{figure=ADS.eps,width=0.16\textwidth} 
\fi
}
\caption{ A directed graph with random rank values associated with its
 nodes. \label{ADS:fig}}
\end{wrapfigure}
 Before continuing to graphs, we introduce some terminology.
For a set $N$ and a numeric function $r:N$,  the function $\kth_r(N)$ returns the $\kth$ smallest value in the
range of $r$ on $N$.  If $|N|<k$ then we define $\kth_r(N)$ to be the
supremum of the range of $r$ (we mostly use $r\in [0,1]$ and the
supremum is $1$.)
  We consider directed or undirected, weighted or unweighted  graphs.
For nodes $i,j$, let $d_{ij}$ be the distance from $i$ to $j$.
For a node $i$ and distance $d$, we use the notation $N_d(i)$ for the set of nodes within distance $d$ from $i$ and
$n_d(i)=|N_d(i)|$ for the cardinality of $N_d(i)$.
We use the notation $\Phi_{<j}(i)$ 
for the set of nodes that are  closer
to node $i$ than node $j$ and 
 $\pi_{ij} = 1+|\Phi_{< j}(i)|$ for the
{\em Dijkstra rank} of $j$ with respect to $i$ ($j$'s position in the
nearest neighbors list of $i$).

 The ADS
of a node $i$, $\ADS(i)$,  is a set of node
 ID and distance pairs.  The included nodes are a sample  of the nodes
reachable from $i$ and with each included node $j\in \ADS(i)$ we store the distance
$d_{ij}$. $\ADS(i)$ is the
union of coordinated MinHash sketches of the neighborhoods $N_d(i)$ (for all possible values of $d$). 
The ADSs are defined with respect to  random
mappings/permutations of the set of all nodes and
come in the same three flavors, according to the underlying MinHash sketches: Bottom-$k$, $k$-mins, and $k$-partition.
For $k=1$, all flavors are equivalent.

For simplicity, our definitions of $\ADS(i)$ assume that
distances $d_{ij}$ are unique for different $j$ (Which can be achieved
using tie breaking).  \notinproc{A definition which does not use tie breaking is
given in Appendix \ref{notieADS:sec}.}

A {\em bottom-$k$} ADS~\cite{bottomk07:ds} is
defined with respect to a single random permutation. 
  ADS$(i)$ includes a node $j$ if and only if
the rank of $j$ is one of the $k$ smallest ranks amongst nodes that are at least
as close to $i$:
\begin{equation} \label{botkADS}
j \in \text{ADS}(i) \iff  r(j) < \kth_r(\Phi_{< j}(i))\ .
\end{equation}

A {\em $k$-partition} ADS (implicit in~\cite{hyperANF:www2011})
is defined with respect to a  random partition $\ppart:V\rightarrow [k]$ of
 the nodes to $k$ subsets $V_h= \{i| \ppart(i)=h\}$ 
and a random permutation.
The ADS
 of $i$ contains $j$ if and only if $j$ has the smallest
 rank among nodes in its bucket that are at least as close to $i$.
\begin{eqnarray*}
\lefteqn{j \in \text{ADS}(i) \iff }\\
&&  r(j) < \min \left\{ r(h) \mid \begin{array}{l} \ppart(h)=\ppart(j)
  \\ \wedge\,  h\in \Phi_{< j}(i) \end{array} \right\}\ .
\end{eqnarray*}

  A $k$-mins ADS \cite{ECohen6f,PGF_ANF:KDD2002}
  is simply $k$ independent bottom-$1$ ADSs, 
  defined with respect to $k$ independent permutations.

 We mentioned that an All-Distances Sketch of $v$
  ``contains'' a MinHash sketch of the neighborhood $N_d(v)$ with respect to any distance $d$.
  We briefly explain how for a given $d \geq 0$, we obtain a MinHash
  sketch of the set of nodes $N_d(v)$ from $\ADS(v)$.  
 Note that the same permutation over
nodes applies to both the ADS and the MinHash sketches.

For a $k$-mins ADS,  the $k$-mins MinHash sketch
of $N_d(v)$ includes, for each of  the $k$ permutations $r$, 
$x \gets \min_{u\in N_d(v)}
r(u)$.  The value for a given permutation is the
minimum rank of a node of distance at most $d$ in the respective
bottom-$1$ ADS.
These $k$ minimum rank values 
$x^{(t)}$ $t\in [k]$ obtained from the different permutations are the  $k$-mins MinHash sketch
of $N_d(v)$. 
Similarly, the  bottom-$k$ MinHash sketch of $N_d(v)$ includes the
$k$ nodes of minimum rank in $N_d(v)$, which are also the $k$ nodes of 
minimum rank in $\ADS(v)$ with distance at most $d$.
A $k$-partition MinHash sketch of $N_d(v)$  is similarly obtained from a
$k$-partition ADS by taking, for each bucket $i\in [k]$,  the smallest
rank value
in $N_d(v)$ of a node in bucket $i$.  This is also the smallest value
in $\ADS(v)$ over nodes in bucket $i$ that have distance at most $d$
from $v$.

 Some of our analysis assumes that the
rank $r(j)$ and (for $k$-partition ADSs) the bucket $\ppart(j)$ are
readily available  for each node $j$.  This can be achieved using 
random hash functions.

    For directed graphs, we consider both the {\em forward} and the {\em
    backward} ADS, which are specified respectively using forward or reverse
paths from $i$.  When needed for clarity,  we augment the notation with
$\overrightarrow{X}$ (forward) and $\overleftarrow{X}$
(backward) when $X$ is the ADS, $N$, or $n$.

\begin{example}
Consider the graph of Figure~\ref{ADS:fig}.
  To determine the forward ADS of node $a$, we sort nodes in
  order of increasing distance from $a$. 
The order is   $a,b,c,d,e,f,g,h$ with respective distances
$(0,8,9,18,19,20,21,26)$.
We now consider the random ranks of nodes to determine which entries
are included in the ADS.
For $k=1$, the (forward) ADS of $a$ includes the lowest rank within
each distance and is:
 $\overrightarrow{\text{ADS}}(a) = \{ (0,a), (9,c),
 (18,d),(26,h) \}$.  
The first value in each pair is the distance from
$a$ and the second is the node ID.  
To compute the reverse ADS of $b$,  we look at nodes in sorted reverse
distance from $b$: $b,a,g,c,h,d,e,f$  with respective reverse
distances $(0,8,18,30,31,39,40,41)$. We obtain
 $\overleftarrow{\text{ADS}}(b) = \{ (0,b), (8,a), (30,c),
 (31,h) \}$.  
The bottom-$2$ forward ADS of $a$ contains
all nodes that have one of the $2$ smallest ranks in the prefix of the
sorted order:  so it also includes 
$\{(8,b),  (20,f)\}$. 
\end{example}

 The expected number of nodes in ADS$(i)$ is $\leq k\ln(n)$, where $n$
 is the number of reachable nodes from $i$:  This was established in
\cite{ECohen6f}   for $k$-mins ADS
and in \cite{bottomk:VLDB2008} for bottom-$k$ ADS.  We present a
proof for completeness.
 \begin{lemma} \cite{ECohen6f,bottomk:VLDB2008} \label{updates:lemma}
The expected size of a bottom-$k$ ADS is
$$k+k(H_n-H_k) \approx k(1+\ln n - \ln k) \ ,$$
where $H_i=\sum_{j=1}^i 1/j$ is the $i$th Harmonic number and $n$ is the number of nodes reachable from $v$.
The expected size of a $k$-partition ADS is accordingly
$k H_{n/k} \approx k (\ln n - \ln k)$.
 \end{lemma}
\begin{proof}
 For bottom-$k$ ADS,  we consider the nodes sorted
by increasing distance from $v$,  assuming unique distances.
  The $i$th node is included
in the bottom-$k$ ADS of $v$ with probability $p_i=\min\{1, k/i\}$.   
Node inclusions are independent (when distances are unique, but
otherwise are negatively correlated).
The expected size of the ADS of $v$ is the sum of node inclusions
which is $$\sum_{i=1}^n p_i = k+k(H_n-H_k)\ .$$
Similarly, for $k$-partition, (assuming a random partition and
permutation), the expected number of included nodes from each bucket is
$\ln (n/k)$ (since each bucket includes in expectation $n/k$ nodes) 
and therefore the total expected size is $k\ln (n/k)$.
\end{proof}

\medskip
\noindent
{\bf Base-$b$ ranks}:  The ADS definition we provided applies with
unique ranks, which are effectively node IDs. Unique IDs are of 
size $\lceil \log_2 n\rceil$ and 
support queries involving meta-data based node selections. 
For many queries, including neighborhood cardinality estimation and
pairwise similarities, it suffices to work with
ranks that have a smaller representation:
For some base $b>1$, 
we use the rounded rank values $r'(j)=b^{-h_j}$,
where $h_j=\lceil -\log_b r(j) \rceil$.  The rounded rank can be
represented by the integer $h_j$.
 The value of the base $b$ 
trades-off the sketch size  and the information it carries, where
both increase when $b$ is closer to $1$.

 With base-$b$ ranks, the expected value of the largest $h_j$, which
 corresponds to the smallest $r(j)$, 
 is $\log_b n$.  Thus, the representation size of the rounded smallest rank
is $\log_2\log_b n$.  The expected deviation from the expectation is
$\leq \log_b 2$, which means that a set of $k$ smallest ranks in a
neighborhood or the $k$ smallest ranks in different permutations can
be compactly represented using an expected number of $\log_2\log_2 n  + k\log_b 2$ bits.

 In the sequel, we consider ranks $r \sim U[0,1]$ and point out
 implications of using base-$b$ ranks. 

\ignore{
\medskip
\noindent
{\bf Exponentially distributed ranks}:
 For estimator analysis, it is sometimes convenient to work with  exponentially
distributed ranks, drawn from an exponential distribution with
parameter 1 instead from the uniform distribution.   We can switch
between these alternatives using a simple
one-to-one monotone transformation $r'(i) = -\ln(1-r(u))$. From
monotonicity, the rank permutation and ADS remains the same.  In
particular, transformation can be applied to rank values after the ADS
is computed.
}

\section{ADS Computation}  \label{comptADS:Sec}

 There are two meta approaches to ADS
computation on graphs.
The first,  {\sc PrunedDijkstra} (Algorithm~\ref{pDijkstra:alg}) uses pruned applications of
Dijkstra's single-source shortest paths
algorithm (BFS when unweighted). A
$k$-mins version was given in \cite{ECohen6f} and
bottom-$k$ version in \cite{bottomk07:ds}.
 The second, {\sc DP},  ($k$-mins in \cite{PGF_ANF:KDD2002} and
 $k$-partition in \cite{hyperANF:www2011}), is node-centric, 
applies to unweighted graphs and is based on
dynamic programming or Bellman-Ford  shortest paths computation.  
Both approaches, however, can be adopted to all three ADS flavors.

  We propose here {\sc LocalUpdates} (Algorithm~\ref{local:alg}),  which extends {\sc DP}
  to weighted graphs.   {\sc LocalUpdates} is 
node-centric and suitable for
MapReduce or similar platforms \cite{pregel:sigmod2010,Naiad:sosp2013}.  
When the node operations are 
synchronized, as with MapReduce, the total number of iterations needed until no more 
updates can be performed is bounded by the diameter of the graph 
(maximum over pairs of nodes of the number of hops in the shortest 
path between them). 

\begin{algorithm}[h]
\caption{$\ADS$ set for $G$ via {\sc PrunedDijkstra}\label{pDijkstra:alg}}
\For {$u$ by increasing $r(u)$}
{
Perform Dijkstra from $u$ on $G^T$ (the transpose graph) \\
\ForEach{scanned node $v$}
{\eIf {$|\{(x,y)\in \ADS(v) \mid y< d_{vu}\}|=k$}{prune Dijkstra at $v$}{$\, \ADS(v)\gets \ADS(v) \cup \{(r(u),d_{vu})\}$}}
}
\end{algorithm}

\ignore{
\begin{algorithm}\caption{$\ADS$ set for $G$ via {\sc PrunedDijkstra}\label{pDijkstra:alg}}
\begin{algorithmic}
\For {$u$ by increasing $r(u)$}
\State Perform a pruned Dijkstra from $u$ on $G^T$ (the transpose graph)
\State When visiting node $v$:
\If {$|\{(x,y)\in \ADS(u) \mid y< d_{vu}\}|$} prune Dijkstra at $v$
\Else {$\, \ADS(v)\gets \ADS(v) \cup \{(r(u),d_{vu})\}$}
\EndIf 
\EndFor 
\end{algorithmic}
\end{algorithm}
}

Both {\sc Pruned Dijkstra's} and {\sc DP} can be performed
in $O(k m \log n)$ time (on unweighted graphs) 
on a single-processor in main memory, 
where $n$ and  $m$ are the number of nodes and edges in the graph.   

 These algorithms maintain a partial ADS for each node, as entries 
of node ID and distance pairs.  $\ADS(v)$ is initialized with the pair
$(v,0)$.     The basic operation we use is {\em edge relaxation}
(named after the corresponding operation in shortest paths computations). 
When relaxing $(v,u)$, $\ADS(v)$ is updated using $\ADS(u)$.  For
bottom-$k$, the relaxation modifies $\ADS(v)$ when $\ADS(u)$ contains a node $h$ such that 
$r(h)$ is smaller than the $k$th smallest rank amongst nodes 
in $\ADS(v)$ with distance at most $d_{uh}+w_{vu}$  from $v$.  More 
precisely,
if $h$ was inserted to $\ADS(u)$ after the most recent relaxation of 
the edge $(v,u)$,
we can update $\ADS(v)$ using {\sc insert}$(v,h,d_{uh}+w_{vu})$:

\SetKwFunction{insert}{insert}
\begin{function} \caption{insert($v,x,a$)}
\KwIn{$\ADS(v)$ and a node-distance pair $(x,a)$}
\KwOut{updated $\ADS(v)$}
\If {$x\not\in \ADS(v)$ and
$r(x) < {\rm k^{th}} \left\{r(y) | y\in \ADS(v) \wedge 
    d_{vy} \leq a\right\}$}
 {$\ADS(v) \gets  \ADS(v) \cup \{(x,a)\}$}
\end{function}

Both {\sc PrunedDijkstra} and {\sc DP} perform relaxations 
in an order which guarantees that inserted entries are part of the 
final ADS, that is, there are no other nodes that are both 
closer and have lower rank: {\sc PrunedDijkstra}
iterates over all nodes in increasing rank, runs
Dijkstra's algorithm from the node on the transpose graph, and prunes at nodes when the ADS is not updated.
{\sc DP} performs iterations, where in
each iteration, all edges $(v,u)$ such that $\ADS(u)$ was 
updated in the previous step,  are relaxed. Therefore, entries 
are inserted by increasing distance. 

The edge relaxation function $\insert$ is stated so that 
it applies for both algorithms,
but some of the conditional statements are redundant:
The test $x\not\in \ADS(v)$ is not needed with {\sc PrunedDijkstra}
(we  only need to record that if/when node $i$ was already updated in the 
current Dijkstra) and the test $d_{vy}<a$ is not needed with DP (since all 
entries in current iteration are of distance at most $a$). 

 To obtain a bound on the number of edge relaxations performed we note that
a relaxation of $(v,u)$ can be useful only when $\ADS(u)$ was modified
since the previous relaxation of $(v,u)$. 
Therefore, each relaxation can be ``charged'' to a modification at its 
sink node, meaning that the total number of relaxations 
with sink $u$ is bounded by the size of $\ADS(u)$ times the in-degree of $u$. 
We obtain that the expected total number of relaxations is 
$O(k m \log n)$.  

 We provide details on bottom-$k$ ADS algorithms.
A $k$-mins ADS set can be computed by performing $k$ separate computations of 
a bottom-$1$ ADS sets (using $k$ different permutations).  To compute 
a $k$-partition ADS set, we perform a separate 
bottom-$1$ ADS computation for each of the $k$ buckets (but with the modification 
that the ADS of nodes not included in the bucket is initialized to $\emptyset$). 
The total number of relaxations is $O(k m \log n)$, which again is $m$ times the expected size of 
of $k$-partition ADS (which is the same as the size of a bottom-$k$
ADS).

{\sc LocalUpdates} incurs more overhead than
{\sc PrunedDijkstra}, as entries can also be deleted from the ADS.  For adversarially constructed graphs (where
distance is inversely correlated with hops), the overhead can be made
linear.
 In practice, however, we can expect a  logarithmic overhead.  We can also
guarantee an $O(\log n)$ overhead by settling for
$(1+\epsilon)$-approximate ADSs (where $\epsilon> 1/n^c$).  A
$(1+\epsilon)$-approximate $\ADS(u)$  satisfies
$$ v\not\in\ADS(u) \implies r(v)> \kth_x\{(x,y)\in
\ADS(u) \mid y<(1+\epsilon)d_{uv} \}\ .$$
We can compute a $(1+\epsilon)$-approximate ADS set by updating
the ADS only when on updates $\insert(i,x,a)$ 
for which the condition is not violated, that is,
$$r(x) < {\rm k^{th}} \left\{r(y) | y\in \text{ADS}(i) \wedge 
    d_{iy} \leq a(1+\epsilon)\right\}\ .$$
We can show that with approximate ADSs, 
the overhead on the total number of updates is bounded by
$\log_{1+\epsilon} \frac{n w_{\max}}{w_{\min}}$, where $w_{\max}$ and $w_{\min}$ are
the largest and smallest edge lengths.  When $\epsilon =
\Omega(1/\text{poly} n)$, we can assume wlog that the ratio
$w_{\max}/w_{\min}$ is polynomial, obtaining a logarithmic overhead.

\begin{algorithm}
\caption{ADS set for $G$ via {\sc LocalUpdates} \label{local:alg}}
\tcp{Initialization}
\For{$u$}{ $\ADS(u) \gets \{(r(u),0)\}$}
\tcp{Propagate updates $(r,d)$ at node $u$}
\If {$(r,d)$ is added to $\ADS(u)$}{
\ForEach{$y \mid (u,y)\in G$}{{\bf send}
 $(r,d+w(u,y))$ to $y$}
}
\tcp{Process update $(r,d)$ {\bf received} at $u$}
 \If {node $u$ {\bf receives} $(r,d)$}
{ \If {$r < \kth_x\{(x,y)\in \ADS(u) \mid y<d \}$}  
 {$\ADS(u) \gets \ADS(u) \cup\{(r(v),d) \}$}
\tcp{Clean-up $\ADS(u)$}
\For{entries $(x,y)\in \ADS(u) \mid y>d$  by 
  increasing $y$}{
\If {$x > \kth_h\{(h,z)\in \ADS(u) \mid z<y \}$}{$\ADS(u) \gets
  \ADS(u) \setminus (x,y)$}
}}
\end{algorithm}
\ignore{
\begin{algorithm}\caption{$\ADS$ set for $G$ via {\sc LocalUpdates}\label{local:alg}}
\begin{algorithmic}
\State {\bf Initialize:} 
\For {$u$} $\ADS(u) \gets \{(r(u),0)\}$
\EndFor 
\State {\bf Send updates:} \If {$(r,d)$ was added to $\ADS(u)$ in 
  the previous iteration} $\forall \{y \mid (u,y)\in G \}$ send 
 $(r,d+w(u,y))$ to $y$. 
\State \EndIf 
\State {\bf Process updates:}
 \If {node $u$ received $(r,d)$}
 \If {$r < \kth_x\{(x,y)\in \ADS(u) \mid y<d \}$}  
\State $\ADS(u) \gets \ADS(u) \cup\{(r(v),d) \}$
\State {\sc Clean-up} $\ADS(v)$ \Comment{Scan entries $(x,y)$ such that $y>d$  by 
  increasing $y$.  Remove $(x,y)$ if $x > \kth_h\{(h,z)\in \ADS(u) \mid z<y \}$}
\EndIf 
\EndIf 
\end{algorithmic}
\end{algorithm}
}

\notinproc{
\subsection{ADS for streams}  \label{streamADS:sec}

  In the introduction we considered ADS computed for a stream of entries of the form
  $(u,t)$, where $u$ is an element from a set and $t$ is a current
  time.  The ADS of distinct elements in a
 stream can be defined with respect to (i) elapsed time from the start to
  the first occurrence of each element or (ii) the elapsed
  time from the most recent occurrence of each element to the current time.
We now describe streaming algorithms which maintain a bottom-$k$ ADS in
both cases. The construction for other flavors is similar.

  The first case is essentially equivalent to 
maintaining a MinHash sketch of the prefix of the
  stream processed so far and recording all entries that resulted in
  the sketch being modified.
We construct an ADS when elements
  are inserted in increasing distance:  For bottom-$k$ ADS we maintain
  $\tau$, which is the $kth$ smallest rank over elements processed so
  far.  When we process an entry $(u,t)$, we test if $r(u)<\tau$, if
  so, we insert $(u,t)$ to the ADS and update $\tau$.  

  The second case was mentioned in \cite{CoSt:pods03f} for streams where
  all entries are distinct.  To sketch only the most recent distinct
  occurrences, we do as follows.
We use some $T$ which is greater than the end
 time of the stream.  Note that entries in the ADS are processed in
 decreasing distance $T-t$.  This means that the newest entry always needs
 to be inserted, and the ADS may need to be cleaned up accordingly.
When processing $(u,t)$, we test if $u$ is already in the ADS and
remove the entry if it is.  We then insert a new entry $(u,T-t)$ and
clean up the ADS by scanning all ADS entries in increasing distance
and removing all entries where the rank is larger than the $k$th
smallest seen so far.
 }

\section{MinHash cardinality estimators}  \label{ADSestNSbasic:sec}
 In this section we 
review estimators for the cardinality $|N|=n$ of a subset $N$ that are applied 
 to a MinHash sketch of $N$.   

The cardinality of $N_d(v)$ can be estimated by
obtaining its MinHash sketch from $\ADS(v)$ and applying a cardinality
estimator to this MinHash sketch.
This also applies to directed graphs, in which case we can estimate the size of
the outbound $d$-neighborhood $\overrightarrow{n}_d(v)$ from
$\overrightarrow{\ADS}(v)$ and similarly estimate the size of the inbound
$d$-neighborhood $\overleftarrow{n}_d(v)$ from
$\overleftarrow{\ADS}(v)$.

As mentioned in the introduction, 
we use the CV to measure the quality of the estimates.
The CV
of an estimator $\hat{n}$
 of $n$ is $\sqrt{\E[(n-\hat{n})^2]}/n$.
  For the same value of the parameter $k$,
the bottom-$k$ sketch contains the most information \cite{bottomk07:ds}, but all flavors 
are similar when $n\gg k$.  We first consider full precision ranks and
then explain the implication of working with base-$b$ ranks.
For illustrative purposes, we start with the 
$k$-mins sketch and then the more informative bottom-$k$ 
sketch.  The lower bound for the $k$-partition 
sketch is implied by the bound for the other flavors.

\subsection{$k$-mins estimator} 
A $k$-mins sketch is a vector $\{x_i\}$ $i\in [k]$.
The
cardinality estimator 
  $\frac{k-1}{\sum_{i=1}^k -\ln (1-x_i)}$ was
presented and analysed in~\cite{ECohen6f}.
It is unbiased for $k>1$.  Its variance
is  bounded only for $k>2$ and the CV is equal to $1/\sqrt{k-2}$.  The
Mean Relative Error (MRE) is
$$\frac{2(k-1)^{k-2}}{(k-2)!\exp(k-1)} \approx
\sqrt{\frac{2}{\pi(k-2)}}\ . $$ 
This estimator can be better understood when we view the ranks as
exponentially distributed with parameter $1$ (rather
than uniform from $U[0,1]$).  These are equivalent, as we can use a simple 1-1 monotone transformation
$y=-\ln(1-x)$ which also preserves the MinHash definition.
In this light,  the minimum rank is exponentially
distributed with parameter $n$.  Our estimation problem is to 
estimate the 
parameter of an exponential distribution from $k$ independent samples
and we use the estimator  $\frac{k-1}{\sum_{i=1}^k
  y_i}$, where $y_i=-\ln(1-x_i)$.

  We now apply estimation theory in order to understand how well
  this estimator uses the information available in the MinHash
  sketch.
\begin{lemma}
Any unbiased estimator applied to the $k$-mins MinHash sketch must have CV that is 
at least $1/\sqrt{k}$. 
\end{lemma}
\begin{proof}
For cardinality $n$, 
each of the $k$ entries (minimum ranks)  is an exponentially 
distributed random variable and therefore has density function $n 
e^{-nx }$.  

Since entries in the $k$-mins sketch are independent, the density 
function (likelihood function) of the sketch is the product of the 
$k$  density functions 
 $f(\boldsymbol{y};n)=n^k e^{-n\sum_{i=1}^k y_i}$. 
Its logarithm, the log likelihood function,  is 
$\ell(\boldsymbol{y};n)= k \ln n -n \sum_{i=1}^k y_i$. 
The Fisher information, $I(n)$, is the negated expectation of the 
second partial derivative of $\ell(\boldsymbol{y};n)$ (with respect to 
the estimated parameter $n$). We have $$\frac{\partial^2 
  \ell(\boldsymbol{y};n)}{\partial^2 n}= -\frac{k}{n^2}\ .$$
This is constant, and equal to its expectation.  Therefore $I(n) = k/n^2$. 

 We now apply the Cram\'er-Rao lower bound which states that the 
 variance of any unbiased estimator is at least the inverse of the 
 Fisher information:
$\frac{1}{I(n)} = \frac{n^2}{k}$.  A corresponding lower bound of 
$\frac{1}{\sqrt{k}}$ on the 
CV is obtained by taking the square root and dividing by $n$. 
\end{proof}

We next show that the sum $\sum_{i=1}^k y_i$ captures all necessary information 
to obtain a minimum variance estimator for $n$.  
\begin{lemma}
The sum of the minimum ranks $\sum_{i=1}^k y_i$ is a sufficient statistics for 
estimating $n$ from a $k$-mins sketch. 
\end{lemma}
\begin{proof}
The likelihood function $f(\boldsymbol{y};n)=n^k 
e^{-n\sum_{i=1}^k y_i}$ depends on the sketch only through the sum $\sum_{i=1}^k y_i$. 
\end{proof}
Therefore, from the Rao-Blackwell Theorem \cite{RaoBlackwell1947}, a minimum variance 
estimator applied to the sketch may only depend on $\sum_{i=1}^k y_i$. 
We can further show that $\sum_{i=1}^k y_i$ is in fact a {\em complete} sufficient 
statistics. 
A sufficient statistics $T$ is complete if any function $g$ for
which $E[g(T)]=0$ for all $n$ must be $0$ almost everywhere (with
probability 1).
The Lehmann-Scheff\'e Theorem \cite{LehSche1950} states that any
unbiased estimator which is a function of a complete sufficient
statistics must be the
unique Uniform Minimum Variance Unbiased Estimator (UMVUE).
Since our estimator is unbiased, it follows that it is the unique UMVUE.
That is, there is no other estimator which is unbiased and has a lower
variance! (for any value of the parameter $n$).

This optimality results provide an interesting insight on
approximate distinct counters.  One can easily come
up with several ways of using the sketch information to obtain an
estimator:  taking the median, averaging quantiles, removing
the two extreme values, and so on.
 The median specifically had been considered
\cite{ECohen6f,ams99,BJKST:random02} because it is convenient for
obtaining concentration bounds.  We now understand that
while these estimators can have variance that is within a constant
factor of optimal, the average carries all information
and anything else is strictly inferior.

\subsection{Bottom-$k$ estimator}
The bottom-$k$ estimator includes the $k$ smallest rank values in $N$,
and we use the estimator
$\frac{k-1}{\tau_k}$,  where
$\tau_k = \kth_r(N)$ is the $k$th smallest rank in $N$. 
This estimator is a conditional inverse-probability
estimator \cite{HT52}:  For each element in $N$ we consider its
probability of being included in the MinHash sketch, conditioned on
fixed ranks of all other elements.  This estimator is therefore unbiased.
The conditioning was applied  with priority sampling \cite{DLT:jacm07}
(bottom-$k$ \cite{bottomk:VLDB2008})  subset sum estimation.

The information content of the bottom-$k$ sketch
is strictly higher than the $k$-mins sketch \cite{bottomk07:ds}.  We
show that the CV of this estimator is upper bounded by the CV of the
$k$-mins estimator:
  \begin{lemma} \label{botkvar:lemma}
The bottom-$k$ estimator has CV $\leq 1/\sqrt{k-2}$.
  \end{lemma}
\begin{proof}
We interpret the bottom-$k$ cardinality estimator
as a sum of $n$ negatively correlated inverse-probability \cite{HT52}  estimates for each
element, which estimate the presence of the element in $N$.  (That is, for
each $v\in N$, estimating its contribution of ``$1$'' to the
cardinality and for each $v\not\in N$, estimating $0$). The inclusion
probability of an element is with respect to fixed ranks of all other
elements.
In this case, an element is 
one of the $k-1$ smallest ranks only if its rank value is 
 strictly smaller than the $k-1$ smallest rank amongst the $n-1$.  For 
 elements currently in the sketch, this threshold value is $\tau_k$.
These estimates (adjusted weights) are equal and positive only for the
$k-1$ elements of smallest rank.  
The variance
of the adjusted weight $a$ of an element conditioned on fixing the rank values of other
elements is $1/p-1$, where $p$ is the
probability that the rank of our element is smaller than the
threshold.  The (unconditional) variance of $a$ is the expectation of
$1/p-1$ over the respective distribution over $p$.

  When ranks are exponentially distributed (which is 
  convenient choice  for  analysis), the distribution of the
  $k-1$ smallest amongst $n-1$ 
is the sum of $k-1$ exponential random variables with
parameters $n-1,n-2,\ldots,n-k+1$.  
We denote the density and CDF functions of this distribution by
$b_{n-1,k-1}$ and $B_{n-1,k-1}$, respectively.
We have $p=1-\exp(-x)$
and the adjusted weight of each element has variance of
$1/p-1= \frac{\exp(-x)}{1-\exp(-x)}$ (conditioned on $x$).  
 We now compute the expectation  of the variance according to the
 distribution of $x$.

 We denote by $s_{n,k}$ and $S_{n,k}$ the respective distribution
function of the sum 
of $k$ exponentially distributed random variables with parameter $n$.

{\small
\begin{align*}
\var[\hat{a}] &= 
\int_{0}^{\infty} b_{n-1,k-1}(x)\frac{e^{-x}}{1-e^{-x}} dx 
\leq \int_{0}^{\infty} s_{n-1,k-1}(x)\frac{1}{x} dx \\
&= 
\int_{0}^{\infty} \frac{(n-1)^{k-1} x^{k-2}}{(k-2)!}
e^{-(n-1)x} \frac{1}{x} dx  \\ &
= \frac{(n-1)^{k-1}}{(k-2)!}
\frac{(k-3)!}{(n-1)^{k-2}} = \frac{n-1}{k-2} \ .
\end{align*}}
The first inequality follows from $\frac{e^{-x}}{1-e^{-x}} \leq 1/x$
and $\forall x, B_{n,k}(x) \leq S_{n,k}(x)$, that is,  $B_{n,k}$ is  dominated by
the sum of $k$ exponential random variables with parameter $n$.
We then substitute the
probability density function~\cite{Feller2} (also used for analyzing
the $k$-mins estimator in \cite{ECohen6f})
$$s_{n,k} = \frac{n^k x^{k-1}}{(k-1)!} e^{-nx}\ .$$
The second to last equality uses $\int_0^\infty  x^a e^{-bx}dx = a!/b^{a+1}$ for natural $a,b$,

Estimates for different elements are negatively correlated (an element being
sampled makes is slightly less likely for another to be sampled) and thus, the variance on the
cardinality estimate is at most the sum of variances of the $n$
elements.  The CV is therefore at most
 $$\sqrt{\frac{n(n-1)}{k-2}}/n\leq \sqrt{\frac{1}{k-2}}\ .$$
\end{proof}
The improvement of bottom-$k$ over the $k$-mins estimator is more pronounced when the cardinality $n$
is smaller and closer to $k$.  The first order term, however, is the
same and when $n\gg k$, the CV of the
bottom-$k$ estimator approaches $\sqrt{\frac{1}{k-2}}$.

 We now consider this estimator from the estimation theoretic lens.
   When $n\leq k$, the variance 
  is clearly $0$.  Therefore, any meaningful lower bound must depend 
  on both $n,k$. 
\begin{lemma}
Any unbiased estimator applied to the bottom-$k$ MinHash sketch must 
satisfy 
$$\lim_{n\rightarrow \infty} \mbox{CV}(n,k) 
\geq 1/\sqrt{k}\ .$$
\end{lemma}
\begin{proof}
Let $x_1,x_2,\ldots,x_k$ be the $k$ smallest ranks in increasing order.   From basic properties of the 
  exponential distribution, the minimum rank 
  $y_0\equiv x_1$ is 
  exponentially distributed with parameter $n$.  For $i>0$,  the difference between 
  the $i$th smallest and the $i-1$th smallest ranks, $y_i\equiv 
  x_{i+1}-x_{i}$,  is exponentially distributed with parameter $n-i$.  
Moreover, these differences $y_i$ are independent. 
We can therefore equivalently represent the information in the 
bottom-$k$ sketch by $(y_0,\ldots,y_{k-1})$, where $y_i$ is 
independently drawn from an exponential distribution with parameter $n-i$.  The joint density function is the 
product 
$f(\boldsymbol{y};n)=\prod_{i=0}^{k-1} (n-i) e^{-(n-i) y_i}$.   The 
Fisher information is $I(n) = \sum_{i=0}^k \frac{1}{(n-i)^2}$.   
We obtain a lower bound on the CV of at least 
$\frac{1}{\sqrt{\sum_{i=0}^{k-1} \frac{n^2}{(n-i)^2}}}$.  When $n\gg k$, the 
expression approaches $\frac{1}{\sqrt{k}}$. 
\end{proof}

\begin{lemma}
$x_k$ (the $k$th smallest rank) is a sufficient statistics for 
estimating $n$ from a bottom-$k$ sketch. 
\end{lemma}
\begin{proof}
We can express the joint density function $f(\boldsymbol{y};n)$ as a product of an 
expression that does not depend on the estimated parameter $n$ and 
$e^{-n\sum_{i=0}^{k-1} y_i} \prod_{i=0}^{k-1} (n-i)$.  Therefore,
$x_k=\sum_{i=0}^{k-1} y_i$ is a sufficient statistics. 
\end{proof}
From Rao-Blackwell  Theorem, the $k$th minimum rank 
captures all the useful information contained in the bottom-$k$ sketch 
for obtaining a minimum variance estimator for $n$.  Since it is a 
complete sufficient statistics, and our estimator is unbiased, it follows from the 
Lehmann-Scheff\'e theorem \cite{LehSche1950} that it is 
the unique UMVUE.

\subsection{$k$-partition estimator}  
 The estimator examines the $1<k'\leq k$ nonempty
buckets, and is conditioned on $k'$.  
The size of each bucket has distribution $1+B(n-k',1/k')$, 
where $B$ is the Binomial distribution.
We can approximate bucket sizes by $n/k'$ and apply the 
$k'$-mins estimator (analysis holds for $k'$ equal buckets).  The estimate is
$\frac{k'(k'-1)}{-\sum_{t=1}^k \ln (1-x^{(t)})}$, where $x^{(t)}$ is the minimum
rank in partition $t$. 

  When $n\gg k$, the $k$-partition estimator performs similarly to the
  bottom-$k$ and $k$-mins estimators.  
When $n < k$, there are effectively only $k'<k$
  nonempty buckets.  Even when $n =O(k)$, the expected size of $k'$ is
significantly smaller than $k$, and the CV is more similar
to that of a $k'$-mins estimate, and therefore, can be expected to be
$\sqrt{k/k'}$ larger.  Moreover, the $k$-partition
 estimator is biased down: 
In particular, when $k'=1$, an event with positive probability, the estimate is $0$.  The probability of $k'=1$
for cardinality $n$ is $p= (1/k)^{n-1}$.  Since we do not
generally know $n$, we can not completely correct for this bias.



\subsection{MinHash sketches with base-$b$ ranks} \label{baseb2:sec}

We considered cardinality estimators for sketches with ``full'' ranks taken from the domain $[0,1]$.  
If we work with truncated  ranks but ensure that there are no rank
collisions, the full-rank estimators can be applied by uniformly at random
``filling in''  missing bits to the desired precision or better yet, computing the
expectation of the estimator over these random completions.
A hash range of size $n^c$ and representation $c\log n$ implies that
  with probability $1/n^{c-1}$ there are no rank collisions between
any two nodes in a set of size $n$.


A uniform random completion of the truncated ranks is an equivalent
replacement to the full rank when all elements with 
the same base-$b$ rank that ``contend''  for a sketch entry are
actually represented in the sketch. If this condition does not hold, the
expected  full-rank completions are more likely to be smaller
than uniform completions and estimates obtained with uniform
completion will be biased.

 To satisfy this condition we need to ensure that there are no rank collisions along the 
``inclusion'' threshold.  With bottom-$k$ this means that the base-$b$ $k$th
smallest rank is strictly  smaller than the base-$b$ $(k+1)$th smallest rank.  With $k$-mins ($k$-partition) it means that the base-$b$ minimum is unique in each permutation (bucket).

If we choose $b=1+1/k^c$, probability of collision is at most $1/k^{c-1}$.
In this case, the expected size of the (integer exponent of the) minimum rank is $\log_2\log_b n
\approx \log_2\log_2 n + \log_2 k^c\approx \log_2\log_2 n  + c\log k$.
Moreover, we recall that the expected size of the offset from this
expectation is constant
times $\log_b 2$.  Substituting $b\approx 1/k^c$ we obtain an expected
offset of the order of $c\log k$, 
so we can compactly represent the
sketch using $\log_2\log_2 n + ck\log k$ bits.

 If we work with a larger base, collisions are more likely and introduce
 bias.  The estimators then need to compensate for the bias.
Specialized estimators for base-$2$ ranks with $k$-mins sketches 
were proposed in \cite{FlajoletMartin85} and for $k$-partition
sketches in \cite{hyperloglog:2007}.  The HIP estimators we 
present next naturally apply with base-$b$ ranks.

\section{The HIP estimator} \label{ADSRCest:sec}

We present the {\em Historic Inverse Probability (HIP)} estimators,
which apply with all three ADS flavors and naturally
extend to base-$b$ ranks.  
We show that the HIP estimators
obtain a factor 2  improvement
 in variance over the respective basic estimator and also show that
 they are asymptotically optimal.
We also present a variant of HIP, the {\em permutation} cardinality
estimator,
which applies to bottom-$k$ ADSs when ranks are a strict permutation of
a domain $[n]$. This estimator improves over plain HIP when the
cardinality is at least $0.2 n$.





For each node $j\in \ADS(i)$ we define its
{\em HIP probability} $\tau_{ij}$ as its inclusion probability in
$\ADS(i)$ conditioned on fixed ranks of all nodes that are closer to
$i$.
The HIP estimator is computed by scanning entries $j\in \ADS(i)$  in order of increasing
distance from  $i$.   For each node $j$, we compute its HIP
probability and then $a_{ij}= 1/\tau_{ij} >0$, which we call the {\em adjusted weight} of $j$.
For $j\not\in \ADS(i)$ we define $a_{ij} \equiv 0$.
The adjusted weight is
an inverse probability 
estimate $a_{ij}= 1/\tau_{ij} >0$ on the presence of $j$ in $\ADS(i)$.
This probability $\tau_{ij}$ is strictly positive for all $j$ reachable
from $i$, therefore, the adjusted weight $a_{ij}\geq 0$ has
expectation $\E[a_{ij}]=1$.  Another important property is that
since $a_{ij}$  is positive if and only if $j\in \text{ADS}(i)$, it can be computed from the information in $\ADS(i)$.
The HIP estimates are a conditioned version of the classic
Horvitz-Thompson \cite{HT52} (inverse probability) estimator.
A similar
conditioning technique, in a different context, was used in \cite{DLT:jacm07,bottomk:VLDB2008}.

 As noted in the introduction, we can estimate $Q_g(i)$ (see
 \eqref{Cquery})
from $\ADS(i)$ using
\begin{equation} \label{Cqueryest}
\hat{Q}_g(i) = \sum_{j} a_{ij}g(j,d_{ij})=\sum_{j\in \text{ADS}(i)} a_{ij} g(j,d_{ij})\ .
\end{equation}
Unbiasedness follows from
linearity of expectation, since each adjusted weight is unbiased.
The second equality holds since only nodes $j\in \ADS(i)$ have
positive $a_{ij}>0$.
We note that the estimate can be easily computed from $\ADS(i)$, since for each
included node $j$ we have the distance $d_{ij}$.


When we are only interested in statistics $Q_g$ where $g(d_{ij})$
  only depends on the distance and not
on the node ID $j$, we can compress the ADS representation to a list
of distance and adjusted weights pairs:  For each unique distance $d$
in $\ADS(i)$  we associate an adjusted weight equal to 
 the sum of the adjusted weights of included nodes in $\ADS(i)$ with distance
$d$. 

To finish the presentation of the HIP estimators, we need to explain
how the adjusted weights are computed for $j\in \ADS(i)$.
We focus in detail on
bottom-$k$ ADSs and start with full ranks $r(i) \sim U[0,1]$.

\subsection{HIP estimate for bottom-$k$ ADS}
Consider a node $v$ and list nodes by increasing Dijkstra rank with
respect to $v$, that is node $i$ has $\pi_{vi}=i$.
We obtain a precise expression for the HIP probability of node $i$:
\begin{lemma}
Conditioned on fixed rank values of all nodes in $\Phi_{<i}(v)$, the
probability of $i\in \ADS(v)$ is
\begin{equation} \label{botkRCthreshold}
\tau_{vi}=
\text{k}^{\text{th}}_r\{ \Phi_{<i}(v) \cap \ADS(v) \}\ .
\end{equation}
\end{lemma}
\begin{proof}
Node $i$ is included if and only if
$ r(i) < \text{k}^{\text{th}}_r\{ \Phi_{<i}(v)\}$, that is, $i$'s rank
is smaller than 
the $k$th smallest rank amongst nodes that are closer to $v$ than
$i$.  Note that it is the same as the $k$th smallest rank among nodes
that are in $\ADS(v)$ and closer to $v$ than $i$, since $\ADS(v)$ must
include all these nodes.  
When $r(i)\sim U[0,1]$, this happens with probability $\tau_{vi}$.
\end{proof}
Note that $\tau_{vi}>0$ happens with probability $1$, which is important
for unbiasedness.
The adjusted weights $a_{vi}$  for node $i$ are $0$ if $i\not\in
\ADS(v)$ and $1/\tau_{vi}$ if $i\in \ADS(v)$.
Note that $\tau_{vi}$, and hence $a_{vi}$, can be computed from $\ADS(v)$ for all $i\in \ADS(v)$.


Note that for $i\leq k$ (when $i$ is one of the $k$ closest nodes to $v$), by definition $i\in \ADS(v)$, $\tau_{vi} \equiv 1$,  and
therefore $a_{vi}=1$, since the first $k$ nodes are included with
probability $1$.  
Also note that the adjusted weights of nodes in $\ADS(v)$ are increasing with
the distance $d_{vi}$ (or Dijkstra rank $\pi_{vi}$).  This is because the inclusion
probability in the ADS decreases with distance.  In particular this
means that the variance of $a_{vi}$ increases with $d_{vi}$.

 We show that the variance of the HIP neighborhood cardinality estimator is at least a
 factor of 2 smaller than the variance of the basic
bottom-$k$  cardinality estimator, which in turn dominates  the basic $k$-mins estimator.
\begin{theorem} \label{adsRCCV:thm}
The CV of the ADS HIP estimator for a neighborhood of size $n$ is
$$\leq \frac{\sqrt{1-\frac{n+k(k-1)}{n^2}}}{\sqrt{2( k-1)}}\leq
\frac{1}{\sqrt{2(k-1)}}\ .$$
\end{theorem}
\begin{proof}
When $n_d(v)\leq k$, the estimate is exact (variance is $0$).
Otherwise, (assuming nodes are listed by Dijkstra ranks
$\pi_{vi}\equiv i$), the variance on  $i$ is $\E[1/p-1]$ where $p$ is the probability that
the rank of $v_i$ is smaller than the $k$th smallest
rank among $v_1,\ldots,v_{i-1}$.
We adapt the analysis of Lemma~\ref{botkvar:lemma} for the variance of
the bottom-$k$ estimator.  We use exponentially distributed ranks, and
have, conditioned on $k$th smallest rank $\tau_{vi}$ in $\Phi_{<i}(v)$ variance
$\exp(-\tau_{vi})/(1-\exp(-\tau_{vi}))$.  We compute the expectation of the variance
for
$\tau_{vi}$ distributed according to $b_{i-1,k}$.
This is a similar computation to the proof of Lemma~\ref{botkvar:lemma} and
we obtain that the variance of the adjusted weight $a_{vi}$ is bounded
by $\frac{i-1}{k-1}$.
Estimates for different $i$ are again negatively correlated and thus the variance of the
neighborhood estimate on $n$ is upper bounded by
$\sum_{i=k+1}^n \frac{i-1}{k-1} = \frac{n^2-n-k^2-k}{2(k-1)}$ and the
upper bound on the CV follows.
\end{proof}

  The bound of Theorem \ref{adsRCCV:thm} extends to distance-decay
closeness centralities. 
\begin{corollary}
For a monotone non-increasing $\alpha(x) \geq 0$ (we define $\alpha(\infty)=0$),
$\hat{C}_{\alpha}(i)= \sum_{j\in \ADS(i)} a_{ij} \alpha(d_{ij})$  is
an unbiased estimator of
$C_\alpha(i) =\sum_j   \alpha(d_{ij})$ with CV that is at most $1/\sqrt{2(k-1)}$.
\end{corollary}
The Corollary holds for the more
  general form \eqref{closenesscdef} when ADSs are computed with
  respect to the node weights $\beta(i)$, see Section
  \ref{nonuniformweights:sec}.
Otherwise, when estimating $Q_g(v)$ using \eqref{Cqueryest}, the variance is
upper bounded as follows:
\begin{corollary}
$$\var[\hat{Q}_g(v)] \leq \sum_{i | v\leadsto i \wedge \pi_{vi}>k}  g(i,d_{vi})^2
\frac{\pi_{vi}-1}{k-1}\ .$$
\end{corollary}
In contrast, we can consider 
the variance of the naive estimator for $Q_g(v)$ that is
mentioned in the introduction.  That estimator uses a MinHash sketch,
which is essentially a random sample of $k$ reachable nodes.  Since
inclusion probabilities are about $\approx k/n$,
where $n$ is the number of reachable nodes from $v$, the
variance in this case is about 
$\frac{n-1}{k-1}\sum_i g(i,d_{vi})^2$.  We can see that 
when $g(i,d_{vi})$ are concentrated (have higher values) on closer nodes, which the MinHash sketch is less likely to include, the variance of the naive estimate can be up to a factor of $n/k$ higher.

\subsection{HIP estimate for $k$-mins and $k$-partition}
  We briefly present the HIP probabilities for $k$-mins and
  $k$-partition ADS.  
 For $k$-mins, a node $i$ is included in
$\ADS(v)$ only if it has rank value strictly smaller than the minimum rank in $\Phi_{<i}(v)$
  in at least one of the $k$ assignments $r_h$ $h\in [k]$.
 Conditioned on fixed ranks of all the nodes $\Phi_{<i}(v)$, we obtain
\begin{equation} \label{kminsRCthreshold}  
\tau_{vi}= 1-\prod_{h=1}^k (1-\min_{j\leq i-1}  r_h(j)) \ .
\end{equation} 
For $k$-partition ADS, we again fix both the rank values and the
partition mapping 
(to one of the $k$ buckets $V_1,\ldots, V_k$) of all nodes in
$\Phi_{<i}(v)$.  We then compute the inclusion threshold, which is the
probability that $i\in \ADS(v)$ given that conditioning.  This is
with respect to a uniform mapping of node $i$  to one of the $k$
buckets and random rank value.  We obtain
\begin{equation} \label{kpartsRCthreshold}
\tau_{vi} = \frac{1}{k}\sum_{h=1}^k
\min_{j\in V_h \cap \Phi_{<i}(v)} r(j)\ ,
\end{equation}
defining the minimum rank
over an empty set $V_h \cap \Phi_{<i}(v)$  to be $1$.

It is simple to verify that the HIP probabilities
are strictly positive when $i$ is reachable from $v$ and that
$\tau_{vi}$, and therefore the respective adjusted
weight $a_{vi}$, can be computed from $\ADS(v)$ when $i\in\ADS(v)$.

\subsection{Lower bound on variance}

  We show that the variance of the HIP estimates is asymptotically
  optimal for $n\gg k$:
 \begin{theorem}  The first order term, as $n\gg k$,  of the
CV of any (unbiased and nonnegative) linear (adjusted-weights based)
estimator of $n_d(v)$ applied to $\ADS(v)$ must be
$\geq 1/\sqrt{2k}$.
\end{theorem}
\begin{proof}
   The inclusion probability of the $i$th node from $v$ $i\geq k$  in a bottom-$k$ ADS$(v)$
is $p_i =k/i$.  If we had known $p_i$,  the best we could
do is use inverse probability weighting, that is, estimate $0$ if not
sampled and $1/p_i$ if the node is included.  The variance of this
ideal estimator is $1/p_i-1$.  There are very weak negative
correlations between inclusions of two nodes, making them almost
independent (for $i\gg k \gg 1$): The probability $p_i$ given that $j<i$ is included is
$\geq (k-1)/i$ and given that $j$ is not included is $k/(i-1)$.  The
covariance is therefore $O(1/i)$.  The sum of all covariances
involving node $i$ is therefore $O(1)$ and the sum of all covariances is
$O(n)$.
The variance of this ideal estimator on a neighborhood of size $n>k $ is
at least the sum of variances minus an upper bound on the sum of covariances
$\var[\hat{n}] = \sum_{i=k+1}^n \frac{i-k}{k}=\frac{(n+k+1)(n-k)}{2k}
  -(n-k) - O(n)$.  The CV, $\sqrt{\var[\hat{n}]}/n$,  has first order
  term for $n\gg k$, of $1/\sqrt{2k}$.

Similar arguments apply to $k$-mins and $k$-partition ADS.  For
$k$-mins ADS, the inclusion probability in ADS$(v)$ of the $i$th node
from $v$ is $p_i= 1-(1-1/i)^k \approx k/i$, and we obtain the same sum for
$i=1,\ldots,n$ as with bottom-$k$ ADS.  For $k$-partition, the
inclusion probability is $p_i= \E[1/(1+x)]$ where $x\sim B[i,1/k]$.
\end{proof}

\subsection{Permutation estimator} \label{permest:sec}

 The permutation estimator we present here is applied to 
a bottom-$k$ ADS that is computed with respect to ranks $\sigma_i \in [n]$
  that constitute a random permutation of  $[n]$.
In terms of information content, permutation ranks dominate
random ranks $r(i) \sim U[0,1]$, since random ranks  can be associated
based on the permutation ranks $\sigma$.
The main advantage of the permutation estimator
is that we obtain tighter estimates when
the cardinality we estimate is a good fraction of $n$.
The permutation estimator is only evaluated experimentally.



  The permutation estimator, similarly to HIP, is viewed  as
  computed over a stream of elements.  In the graph setting, the stream corresponds to
  scanning of nodes so that first occurrences of nodes are
  according to  increasing distance from $v$.
The entries in $\ADS(v)$ correspond to nodes on which the sketch was updated. A
positive weight is then associated with these updates.
The weight is an estimate of the number of 
 distinct elements scanned from the previous update (or the beginning if it 
 is the first update) to the current one.  We maintain a running
 estimate $\hat{s}$ on the cardinality $s$ of the set $S$ of distinct
 elements seen so far.  When there is an update, $\hat{s}$ is
 increased by its weight $w$.

The first $k$ updates corresponds to the first $k$ distinct elements. 
Each of these updates has weight $1$ and when the cardinality $s\leq 
k$, our estimate is exact  $\hat{s} = s$.

 Consider now an update that occur after the first $k$ distinct
 elements.  Let $\mu>k$ be the
$k$th smallest rank in $S$ (which is the $k$th smallest permutation
rank in the bottom-$k$ sketch).

We now argue that after an update, the expected number of distinct
nodes until we encounter the next update is 
$w'=\frac{n-s+1}{\mu-k+1}$.     To see this, note that
the number of nodes in $S$ with permutation rank $\mu$ or below
is $k$.  So there are $\mu-k$ remaining
nodes with rank smaller than $\mu$ amongst those in $[n]\setminus S$.
The expectation is that of sampling without replacement until we find
a node with permutation rank below $\mu$.

 When the update occurs, we would like to compute $w'$ and update our
 estimate $\hat{s}$.
But we actually do not know $s$. So  instead we plug-in the unbiased estimate
 $\hat{s}$ to obtain
$w= \frac{n-\hat{s}+1}{\mu-k+1}$.  We then update the bottom-$k$
sketch (and $\mu$ if needed) and $\hat{s}\gets \hat{s}+w$.

   Note that when $\mu=k$, that is, the $k$ smallest elements of the permutation, those
   with
$\sigma_i\leq k$,  are  included in $S$, the probability of an update
is $0$ as the sketch is saturated.  We then need to correct the
estimate to account for the number of nodes that are farther than the
nodes with permutation rank $[k]$.   The correction is computed as follows.

  If the cardinality is $x$, then conditioned on it including all
  the  elements with permutation ranks $[k]$, the expected number of elements
  that are farther than all the elements with permutation ranks in $[k]$ is
 $\frac{x-k}{k+1}$.  So the expected number of elements till the last update 
is $x'=x-\frac{n-k}{k+1}$.   Note that our estimate $\hat{s}$ was
unbiased for $x'$.

Solving $x-\frac{x-k}{k+1}=x'$ for $x$ we obtain the relation
$x=x'\frac{k+1}{k}-1$.  We plug-in $\hat{s}$ for $x'$ and obtain the correction
$\hat{x}= \hat{s}\frac{k+1}{k}-1$.  This correction is used when our
sketch
contains the $k$ elements of permutation ranks $[k]$.

\subsection{Simulations}

We use simulations to study the Normalized Root Mean Square Error
(NRMSE), which corresponds to the CV
when estimator is unbiased,
and the {\em Mean Relative Error} (MRE), defined as
$\E[|n-\hat{n}|]/n$ of the basic, HIP, and permutation  neighborhood-cardinality
estimators. 
  We evaluated
the basic estimators for all three flavors and the bottom-$k$ HIP
estimators. We use
sketches with full ranks, because the optimal basic estimators
are well understood with full ranks. 
Actual representation size for ``full'' ranks is discussed in Section \ref{baseb2:sec}.

 The cardinality $n_d(v)$ is estimated from nodes in $\ADS(v)$ of
 distance at most $d$.   The structure of the ADS and the
behavior of the estimator as a function of the cardinality $n_d(v)$ do
  not depend on the graph structure.   
When nodes are presented in increasing distance 
from $v$, the ADS only depends on the ranks assigned to these nodes.
Our simulation is therefore
  equivalently performed on a stream of $n$ distinct elements, and ADS
  content
is built from the randomized ranks assigned to these elements.
After processing $i$ distinct element, we obtain an estimate of $i$
from the current ADS.  We do so for each cardinality.  We use
multiple runs of the simulation, which are obtained by different randomization of
ranks.
In case of the permutation estimator, the ranks we use are permutation
ranks from a random permutation on all $n$ nodes.  For other estimators,  the
estimate for a certain cardinality does not depend on the total number of nodes.

  Figure~\ref{nsizeest:fig} shows the NRMSE  and the MRE estimates by
average of multiple simulation runs.  We also provide,
for reference, the exact values of the CV ($1/\sqrt{k-2}$)
and MRE ($\approx\sqrt{2/(\pi(k-2))}$) of the $k$-mins basic estimator.  These values are
independent of cardinality and upper
bounds the respective measures for the basic bottom-$k$ estimator.

 Looking at basic estimators,
we can see that (as expected from analysis) for $n\gg k$, the error is
similar for all three flavors and the NRMSE is around $1/\sqrt{k-2}$.
For smaller values of $n$, the bottom-$k$
estimator is more accurate than the $k$-mins which in turn is more
accurate than the $k$-partition estimator:  The bottom-$k$ estimator
is exact for $k\leq n$ and then the relative error slowly increases
until it meets the $k$-mins error.
  We can observe that, as explained by analysis, the
$k$-partition estimator  is less accurate for $n\leq 2k$.

  The figures also include the first-order term (upper bound) for
HIP.  The results for the bottom-$k$ HIP estimator 
clearly demonstrate the improvement of the HIP
estimators: We can see that the error of the
bottom-$k$ HIP  estimator is a factor of $\sqrt{2}$ smaller than that
of the basic bottom-$k$ estimator.   
The figures also demonstrate the benefit of using our permutation
estimator:  The NRMSE and MRE of the permutation estimate were always
at most that of HIP.  The two are comparable when the
estimated cardinality is at most $0.2 n$.  When it exceeds $0.2 n$,  we observe a
significant advantage for the permutation estimator over plain HIP.


\begin{figure}[htbp]
\centerline{
\begin{tabular}{cc}
\ifpdf
\includegraphics[width=0.21\textwidth]{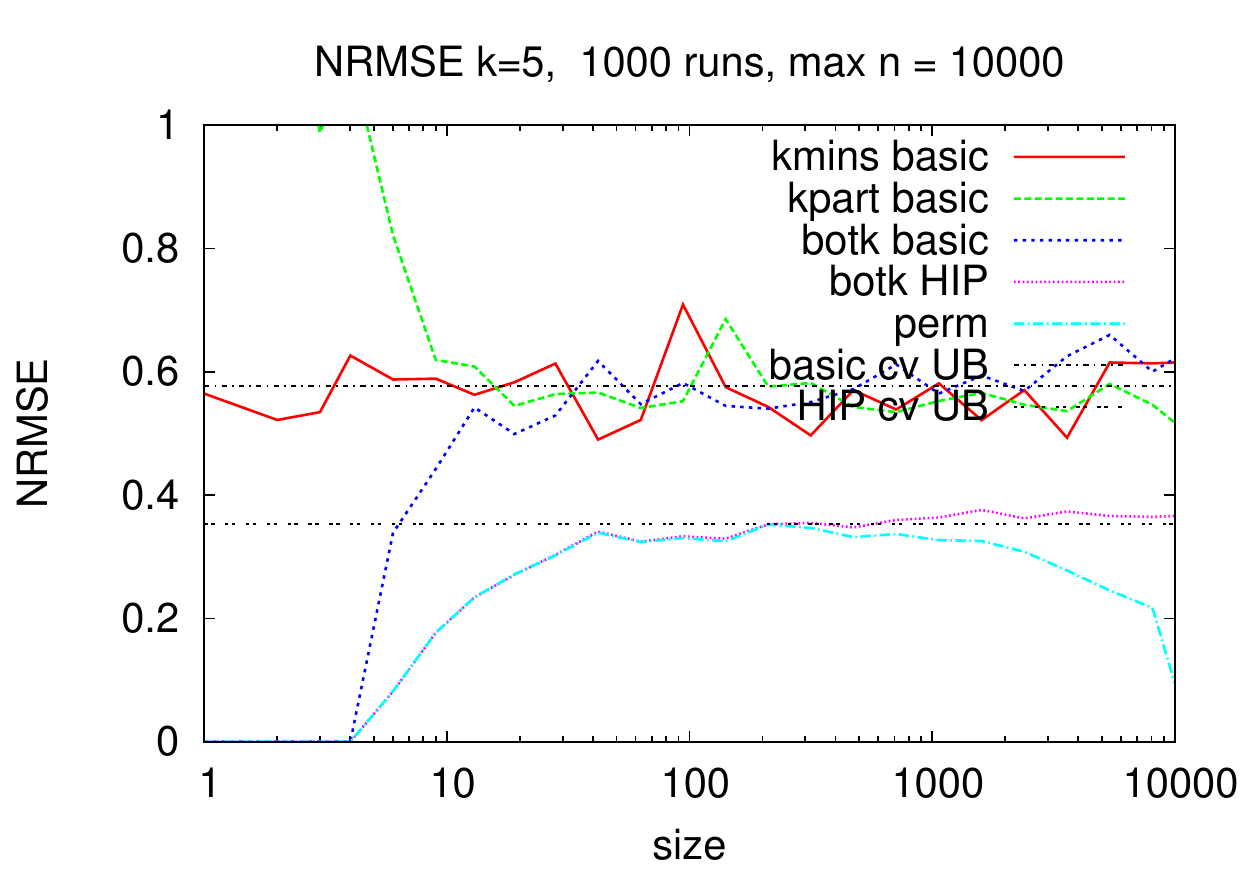} &
\includegraphics[width=0.21\textwidth]{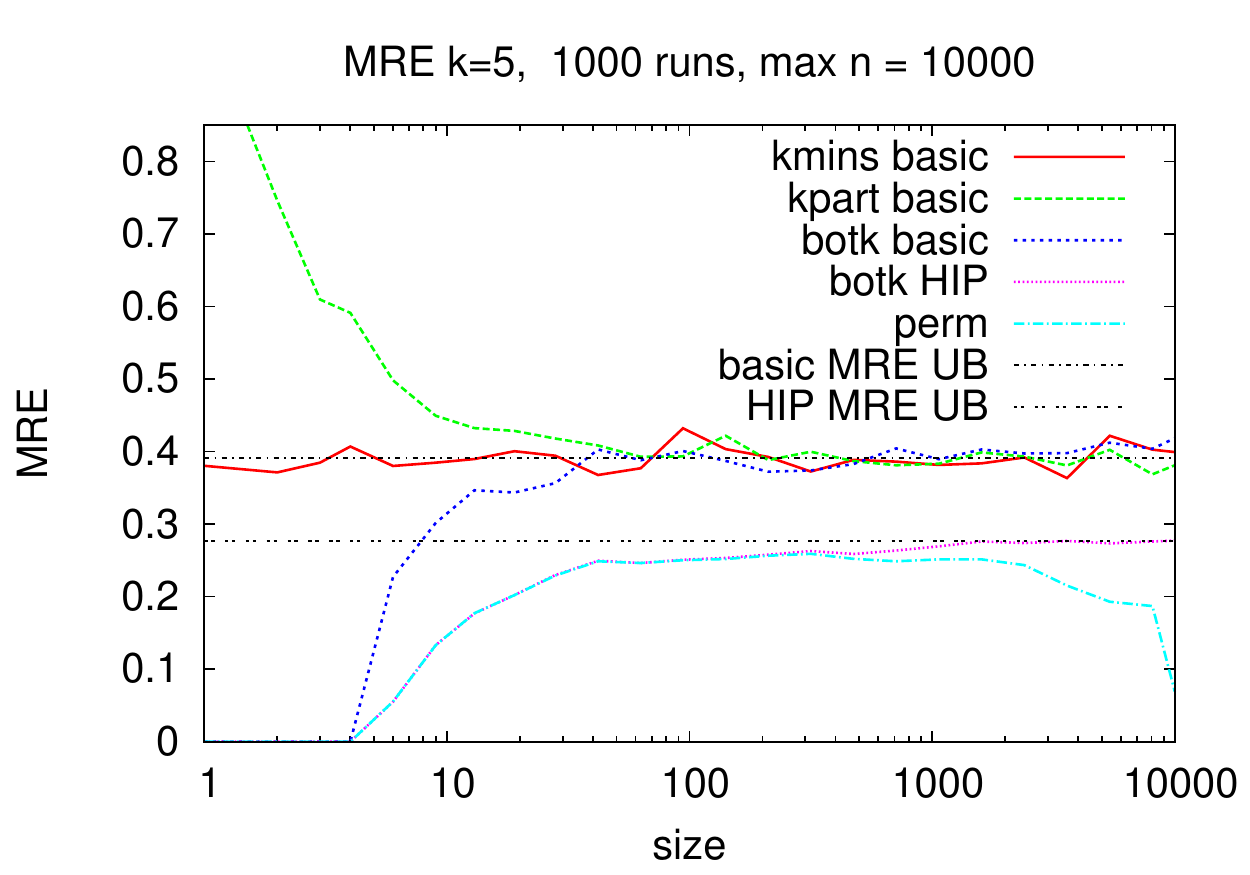} \\
\includegraphics[width=0.21\textwidth]{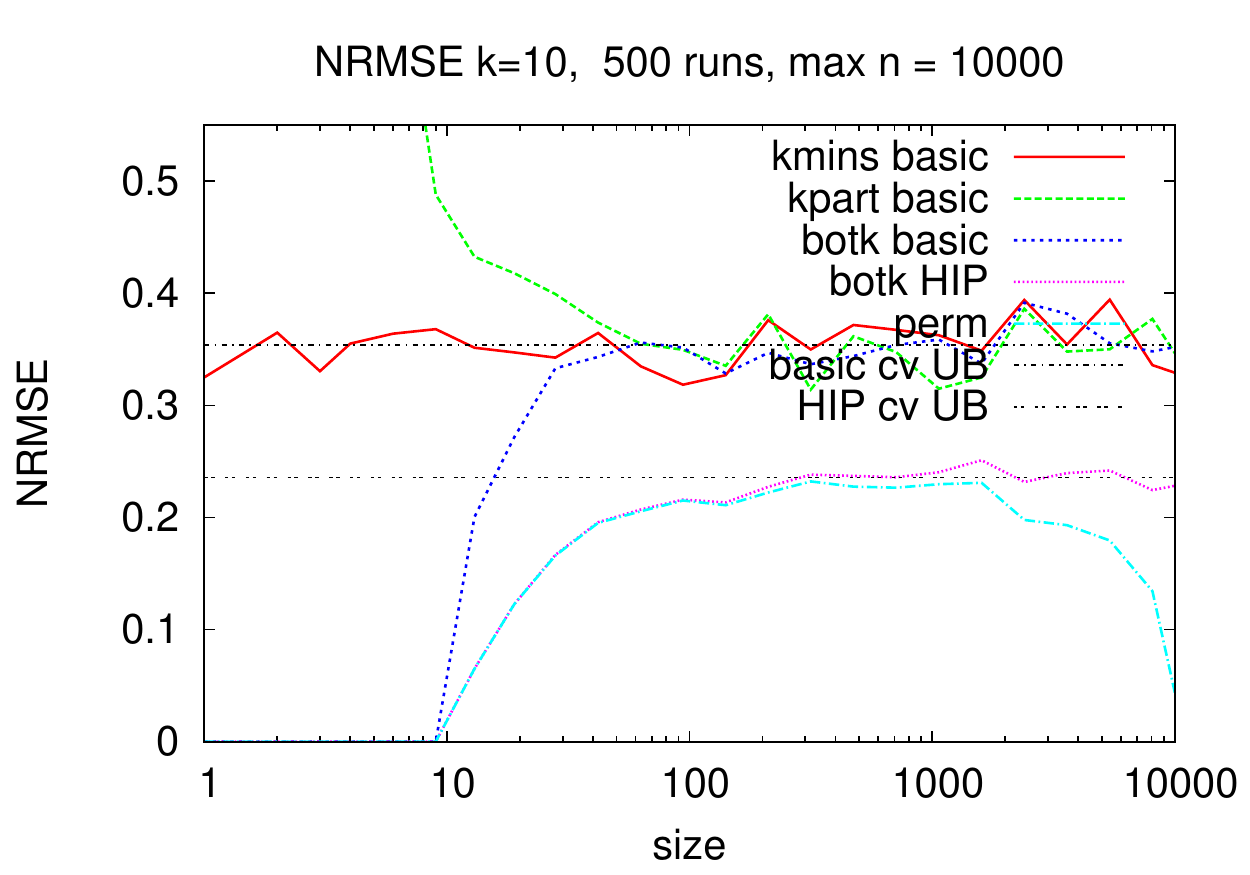} &
\includegraphics[width=0.21\textwidth]{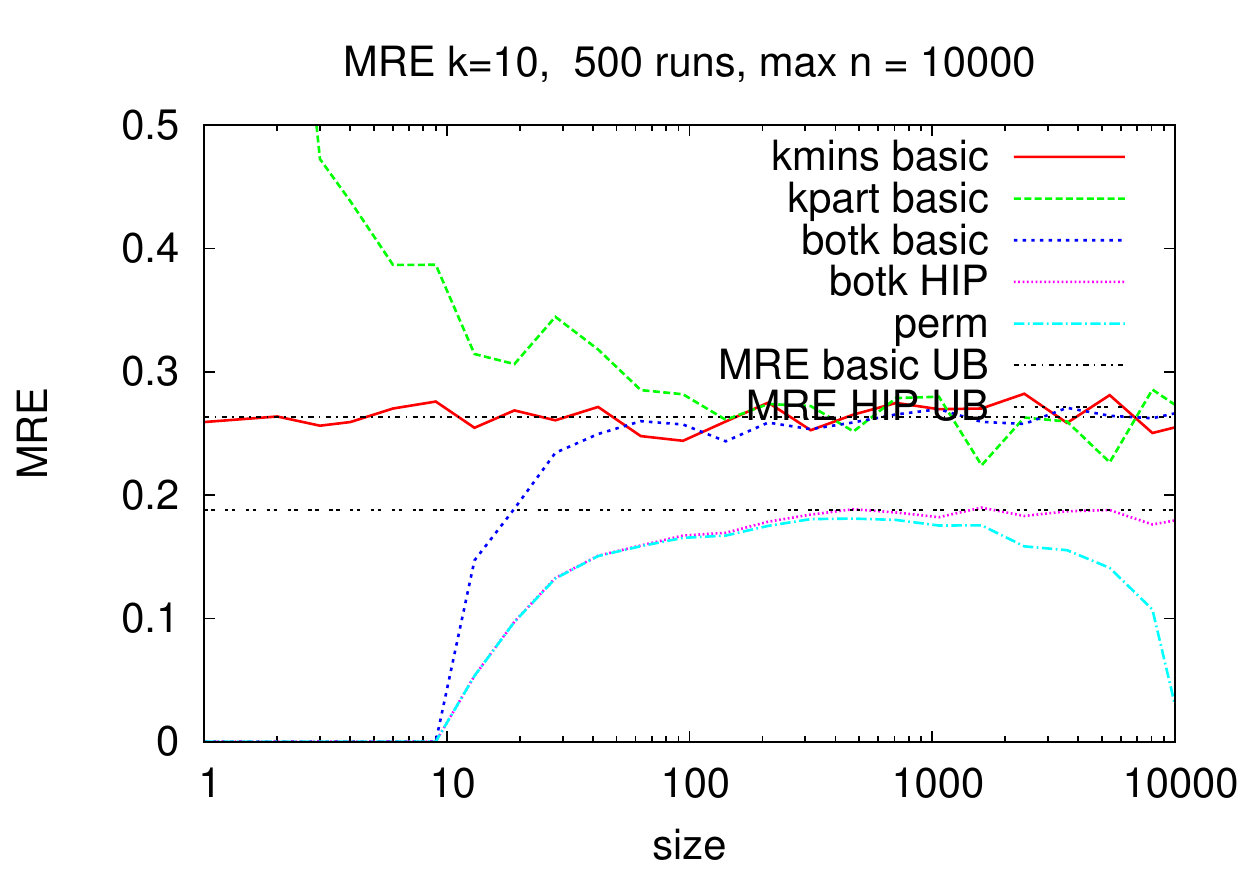} \\
\includegraphics[width=0.21\textwidth]{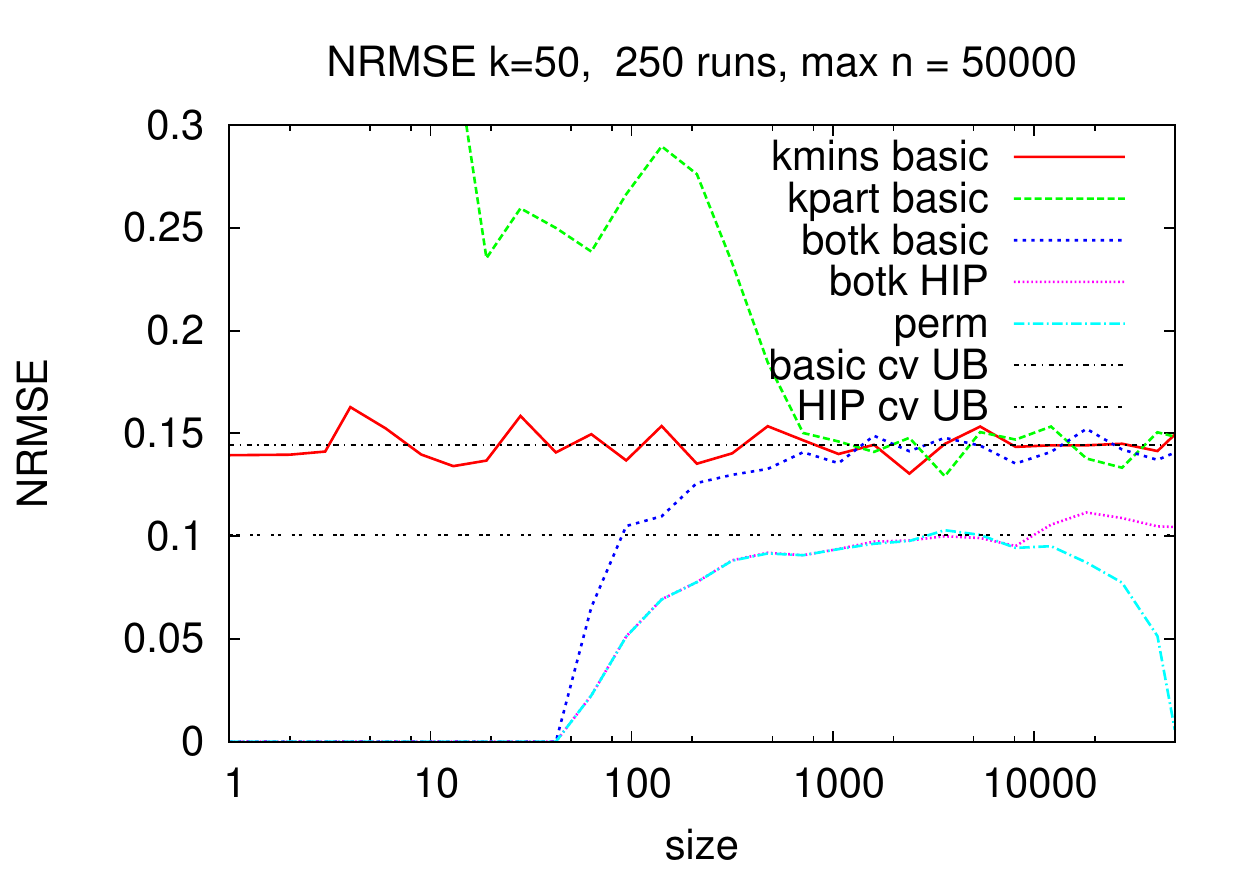} &
\includegraphics[width=0.21\textwidth]{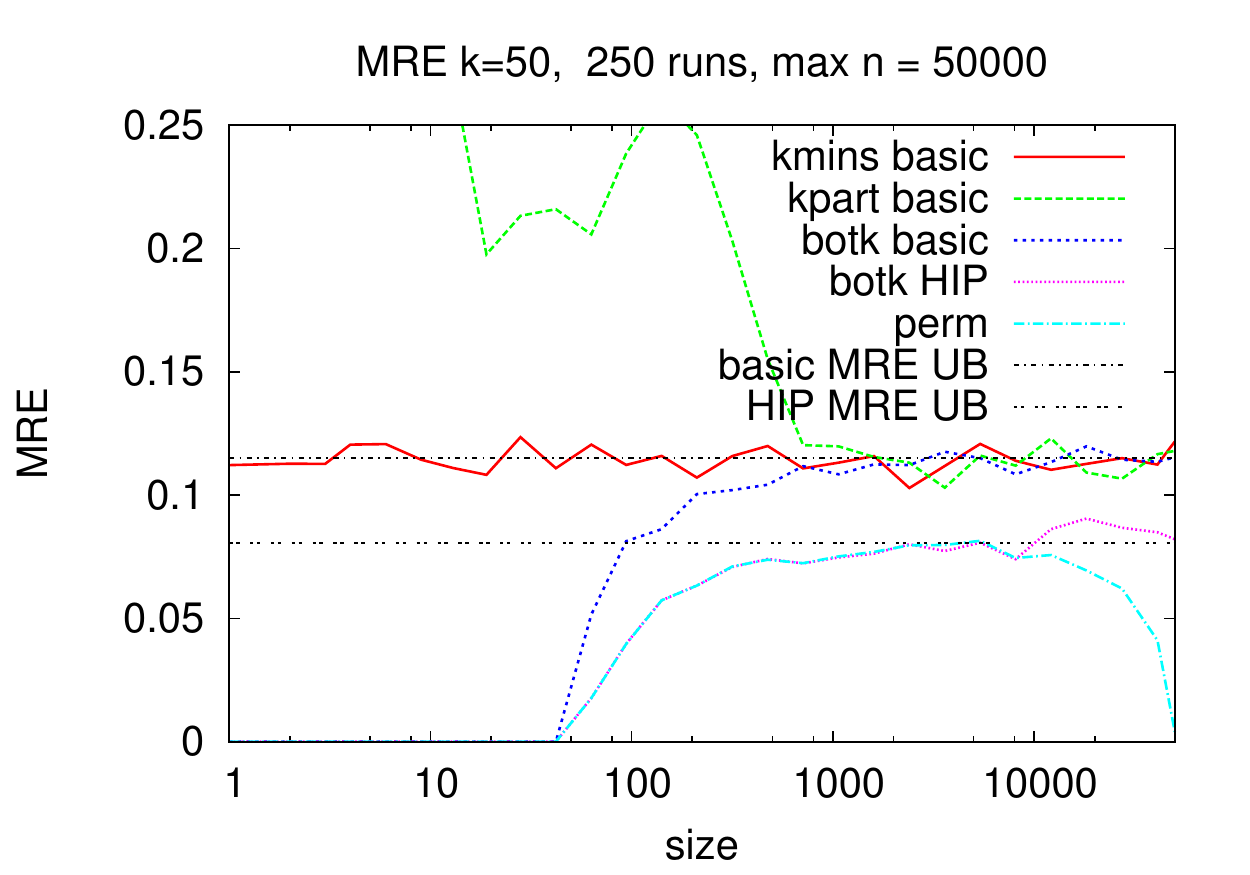} \\
\else
 \epsfig{figure=resNRMSE10000k5r1000.eps,width=0.21\textwidth} &
 \epsfig{figure=resMRE10000k5r1000.eps,width=0.21\textwidth} \\
 \epsfig{figure=resNRMSE10000k10r500.eps,width=0.21\textwidth} &
 \epsfig{figure=resMRE10000k10r500.eps,width=0.21\textwidth} \\
 \epsfig{figure=resNRMSE50000k50r250.eps,width=0.21\textwidth} &
 \epsfig{figure=resMRE50000k50r250.eps,width=0.21\textwidth} \\
\fi
\end{tabular}
}
\caption{ NRMSE (normalized root mean square error) and MRE (mean
  relative error) of neighborhood size estimators with $k=5,10,50$, as
  a function of neighborhood size, averaged over multiple runs.
We show $k$-mins, bottom-$k$, and $k$-partition basic estimators and
our bottom-$k$ HIP and permutation estimators.  For reference, we
also show the exact values $1/\sqrt{k-2}$ and $1/\sqrt{2(k-1)}$ of the CV  of
the basic and HIP $k$-mins estimators.  These are upper bounds on the
CV of respective bottom-$k$ estimators.  We also show
$\sqrt{\frac{2}{\pi(k-2)}}$ for the MRE of the basic $k$-mins
estimator and $\sqrt{\frac{1}{\pi(k-1)}}$ as a reference MRE for HIP.
 \label{nsizeest:fig}}
\end{figure}

\subsection{HIP with base-$b$ ranks}  \label{limprecRC:sec} 

The application and analysis of HIP estimators carries over,
retaining unbiasedness even with
collisions.  To do so we need to compute the HIP probability $\tau_{vi}$, which depends on the
  ranks of all closer nodes, taking the discretization into account
  (instead of using the full-rank expressions).



 When using base-$b$ ranks, however, the HIP
probability $\tau_{vi}$ is a ``rounded down'' form of
the corresponding full rank probability.   Since the probability is
strictly smaller than with the full ranks, the contribution to the variance of the estimate, which is
$1/\tau_{vi}-1$, is higher.
We perform a back-of-the-envelope calculation which shows that $\tau_{vi}$ can be expected to
increase by a factor of  $\frac{1+b}{2}$, which implies the
same-factor increase in variance:
 Considering a range between discretized values, $a=1/b^i$ and $ba=1/b^{i-1}$, and assuming the full rank $x$ lies uniformly in that interval.
The full-rank inclusion probability is $x$ whereas the rounded-down one is $a$.
We consider the expectation of the ratio $x/a$.  This expectation is
$$\frac{1}{a(b-1)}\int_a^{ba} \frac{x}{a} dx = \frac{b+1}{2}\ .$$

 Simulations in the next section show that this calculation is fairly
 accurate.   We can use this calculation 
to find a sweet spot for the base $b$, considering the tradeoff between
representation size  and variance.
The CV as a function of $k,b$ is
$\sqrt{\frac{(1+b)}{4 (k-1)}}$.  The representation size depends on
application.  If sketch is only used for counting, maintaining few
bits for counter, there is diminishing value with smaller bases.  If
the sketch is used as a
sample which supports segment statistics (selection) and stores meta-data (or
node IDs, then $k$ is the dominant term and it is beneficial to work
with full ranks.

\section{Approximate distinct counting} \label{distinctcount:sec}
\begin{figure}[htbp]
\centerline{
\begin{tabular}{cc}
\ifpdf
\includegraphics[width=0.21\textwidth]{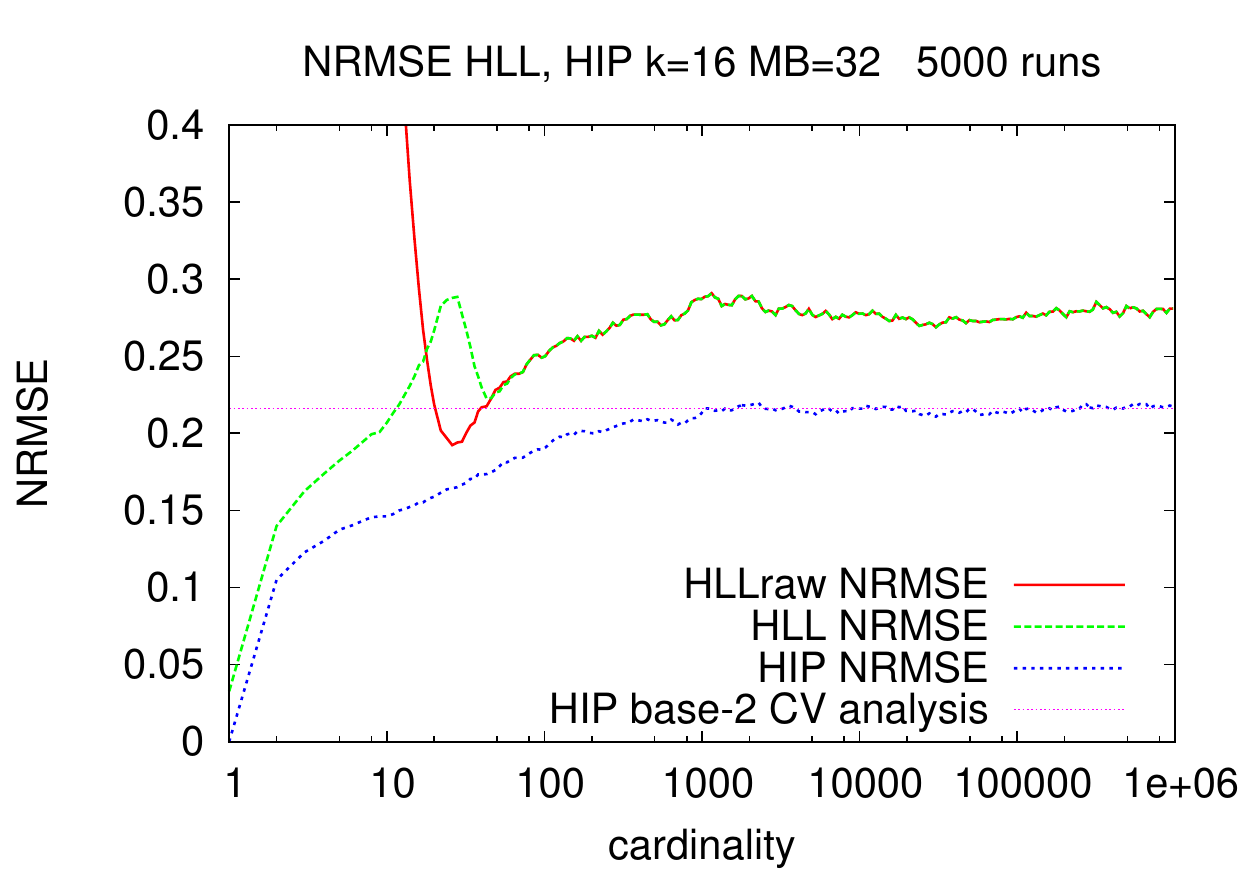} &
\includegraphics[width=0.21\textwidth]{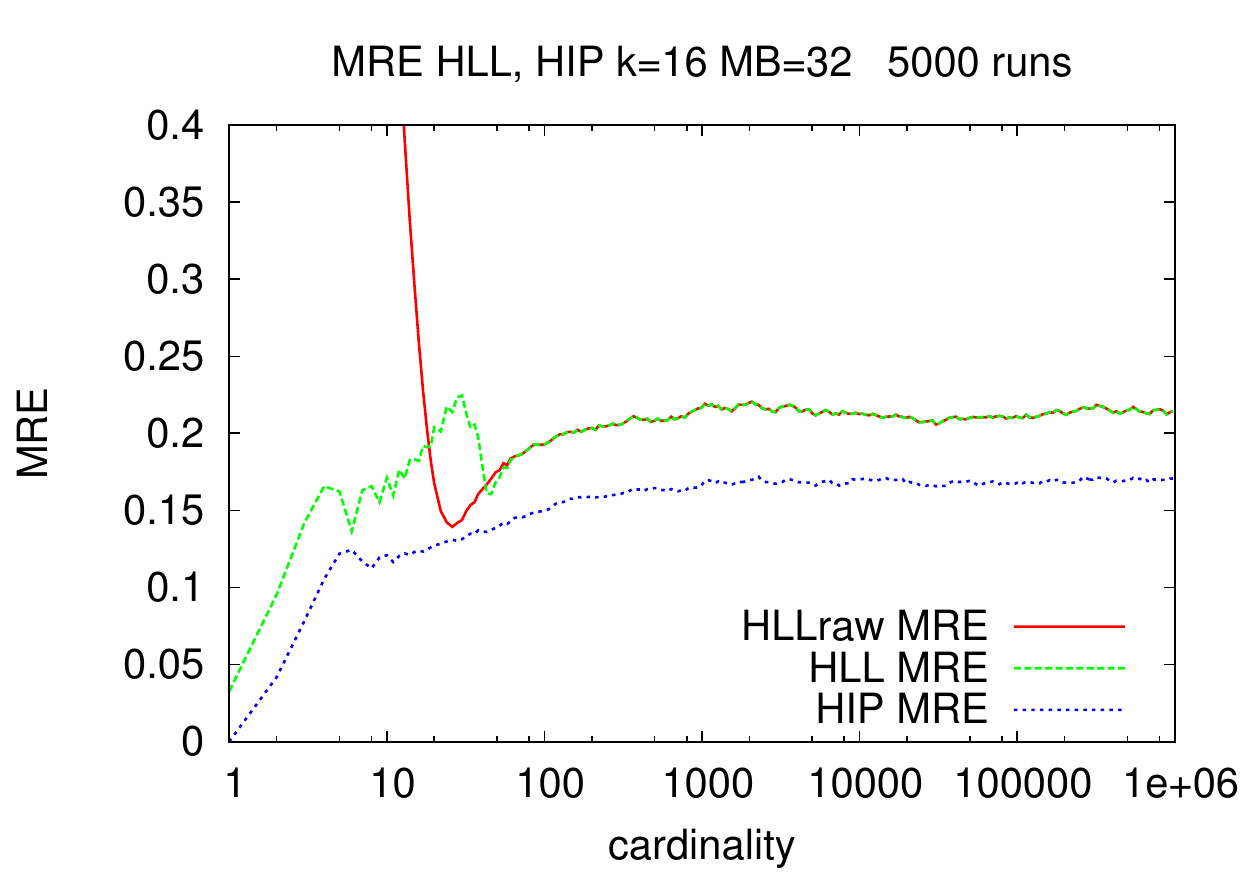} \\
\includegraphics[width=0.21\textwidth]{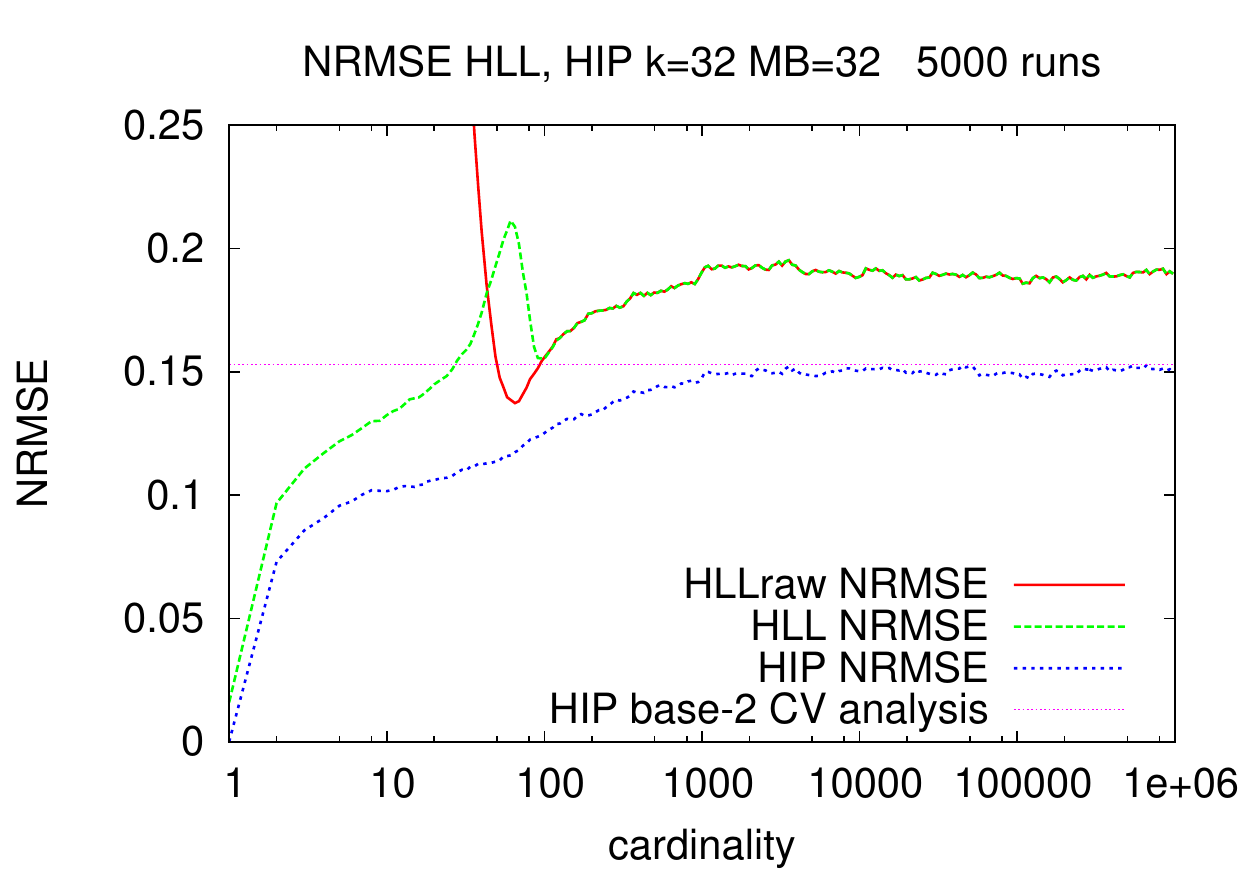} &
\includegraphics[width=0.21\textwidth]{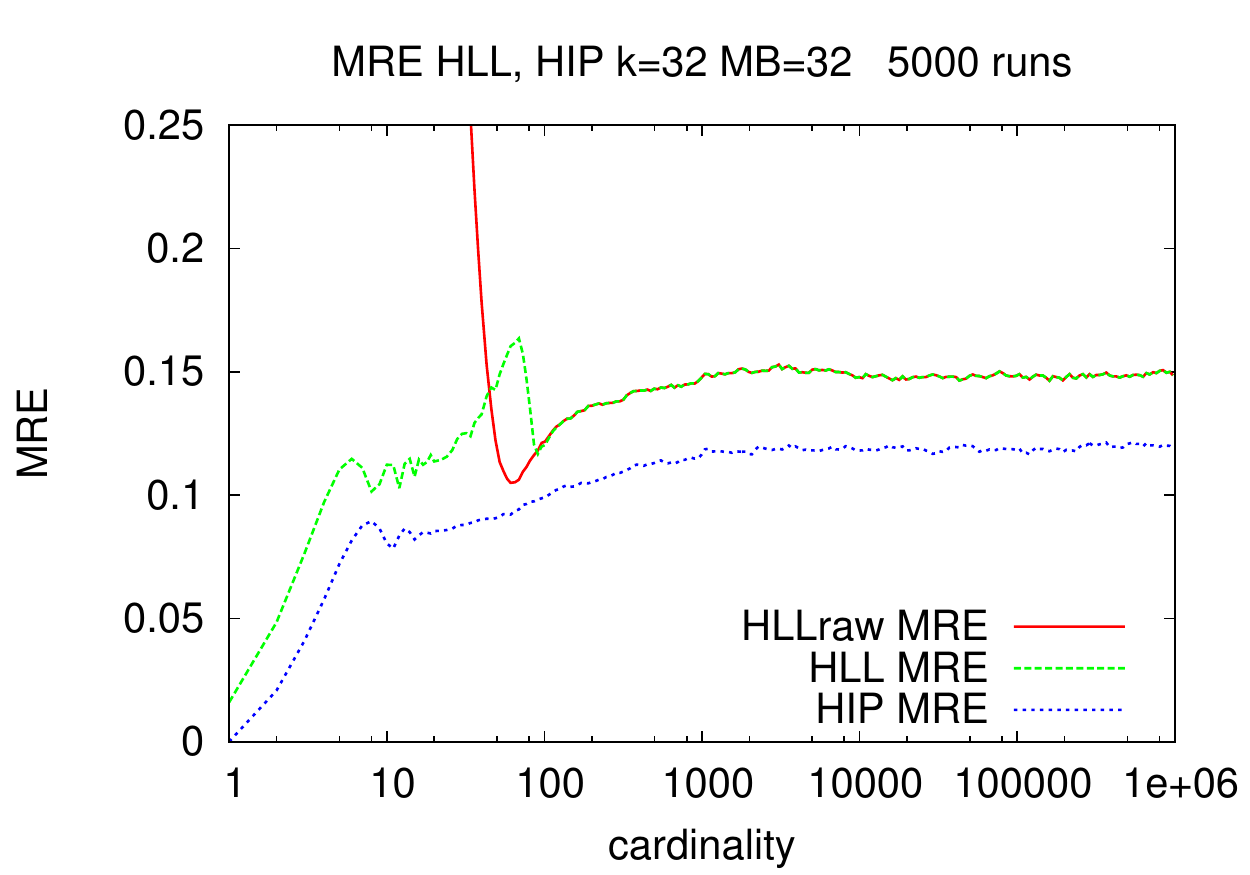}  \\
\includegraphics[width=0.21\textwidth]{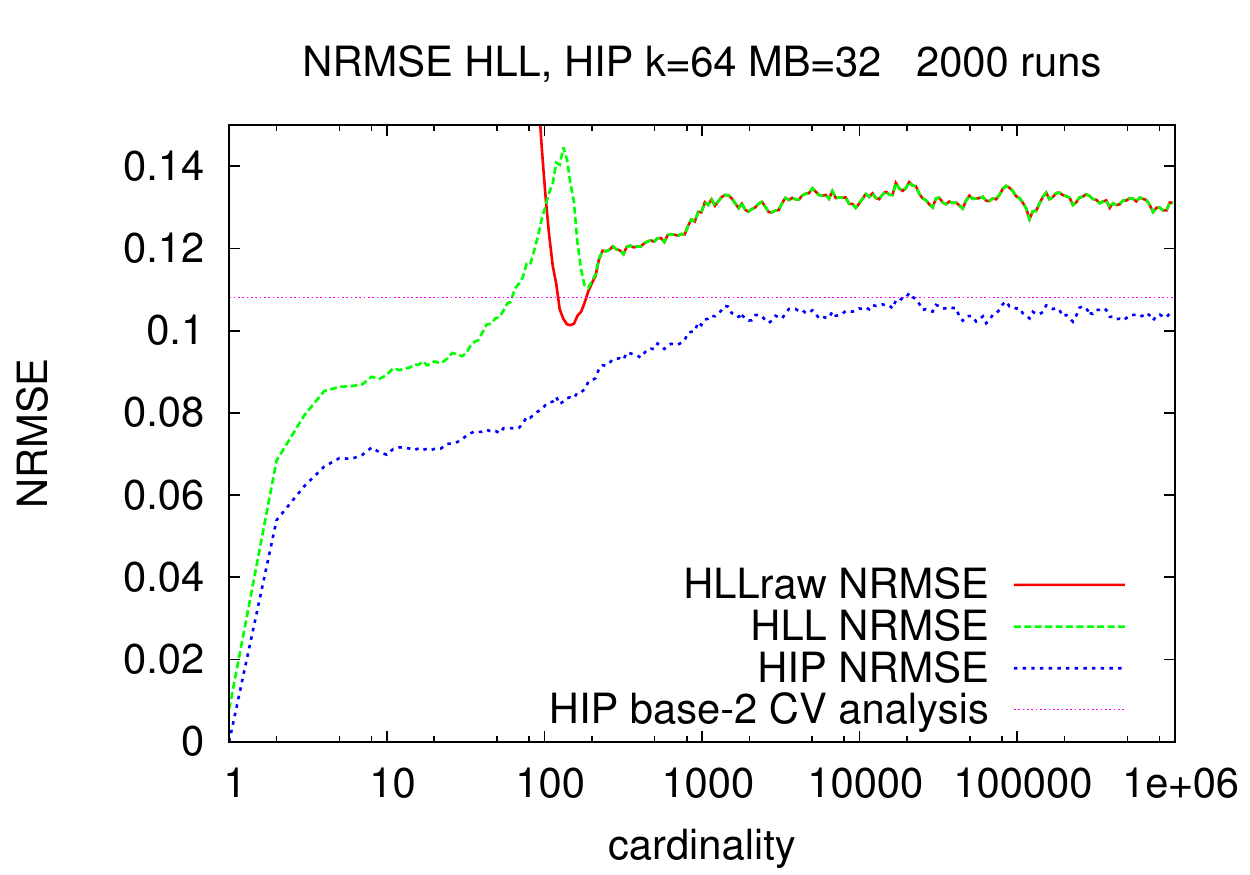}  &
\includegraphics[width=0.21\textwidth]{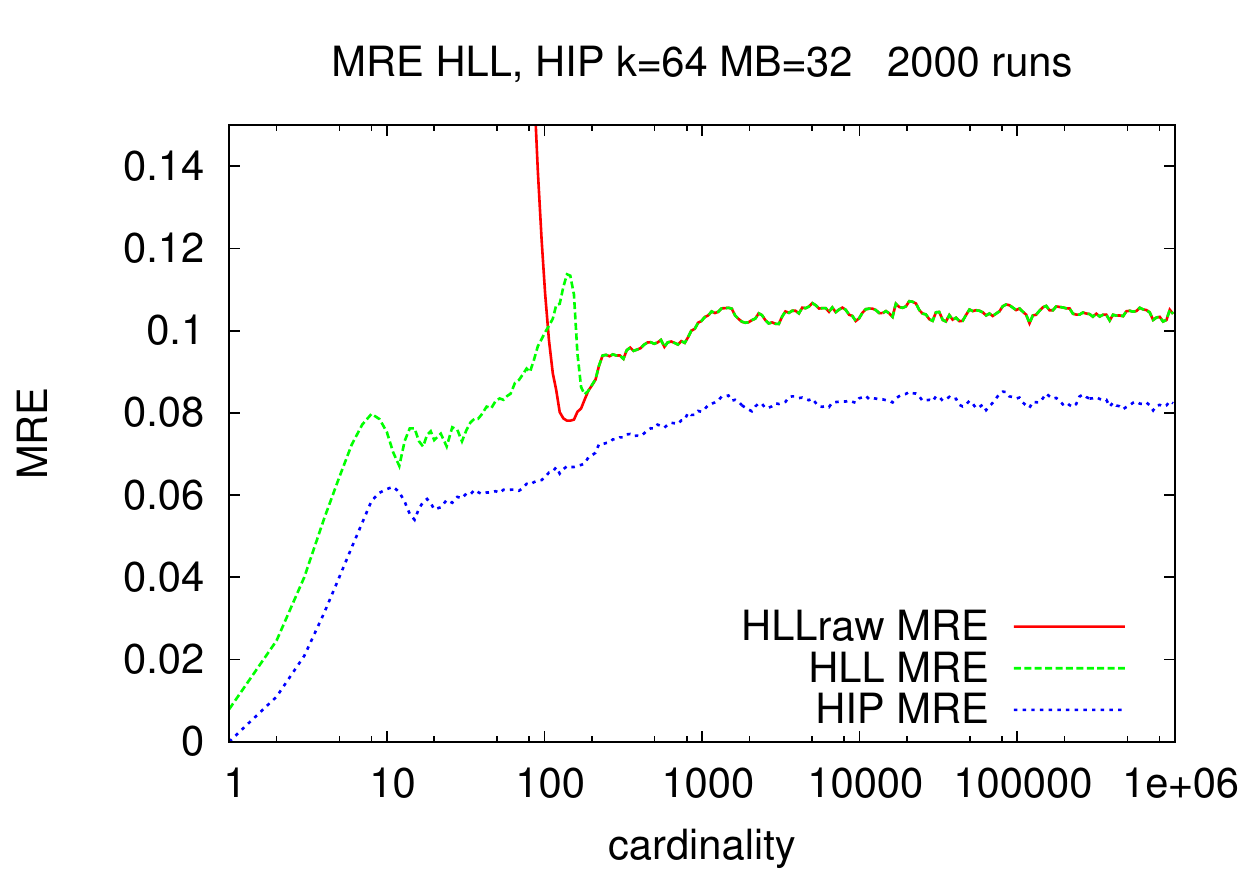} \\
\else
 \epsfig{figure=res_HLL_HIP_k16_r5000_MSE.eps,width=0.21\textwidth} &
 \epsfig{figure=res_HLL_HIP_k16_r5000_MRE.eps,width=0.21\textwidth} \\
 \epsfig{figure=res_HLL_HIP_k32_r5000_MSE.eps,width=0.21\textwidth} &
 \epsfig{figure=res_HLL_HIP_k32_r5000_MRE.eps,width=0.21\textwidth} \\
 \epsfig{figure=res_HLL_HIP_k64_r2000_MSE.eps,width=0.21\textwidth} &
 \epsfig{figure=res_HLL_HIP_k64_r200_MRE0.eps,width=0.21\textwidth} \\
\fi
\end{tabular}
}
\caption{HIP and HLL raw and bias-corrected estimators.  Applied with $k=16,32,64$
and 5-bit counters.  (k-partition base-2 MinHash sketches) \label{HLLHIP:fig}}
\end{figure}

Our HIP estimators (with full or base-$b$ ranks) can be used to approximate
the number of distinct elements in a stream.   
A MinHash sketch is maintained for the
distinct elements on the prefix of the stream that is processed.


To apply HIP, we augment the MinHash sketch with an additional
register
which maintains an 
 approximate count of the number of distinct 
elements.    
Each time the sketch is updated, we compute the adjusted weight of the 
element 
and accordingly increase the count by that amount. 
Since  the expected total number of updates is $\leq k\ln 
n$, where $n$ is the number of distinct elements in the stream (see 
Lemma \ref{updates:lemma}), 
the additional work performed for an update balances out as 
a diminishing fraction of the total stream computation.

An explicit representation of the additional counter as an approximate
counter (see Section~\ref{approxcount:sec})
would require storing the exponent, which is of size $\lceil \log\log
n\rceil+1$ and $\lceil \log_2 \sqrt{(4k/3)}+4\rceil$ significant bits
(precision with respect to the CV of HIP).  The exponent can more
efficiently be stored as an offset to the exponent values stored in
the sketch, removing its dependence on $n$.  Thus, using only $O(\log_2 k)$
for the approximate count.  An even more compact representation
of the approximate count also eliminates the dependence on $k$, and
requires only few bits in total.  To do that we represent the HIP estimate as a
correction of a basic estimate obtained from the MinHash sketch. The
correction can be expressed as a signed multiplier of
$\hat{n}/\sqrt{k}$ using a fixed number of bits.  When the sketch is
updated, we recompute the basic estimate $\hat{n}$ and accordingly
update the correction to be with respect to the new HIP estimate.

  HIP is very flexible.  It applies to all three sketch flavors, to
  full and to base-$b$ ranks, and also works with truncated registers that can
  get saturated.  In this case we simply take the update probability
  of a saturated register to be $0$.  The HIP estimate quality
  gracefully degrades with the number of saturated registers.
  Eventually, if all registers are saturated, the HIP estimate saturates
 and becomes biased.
In order to compare the HIP estimate with HyperLogLog (HLL)
\cite{hyperloglog:2007,hyperloglogpractice:EDBT2013}, which is the
state of the art approximate
distinct counter, we implemented it on the same MinHash sketch that
is used by HLL.  HLL uses $k$-partition MinHash sketches with base-2 ranks.  The
registers have $5$ bits and are thus saturated at 31.  
Pseudo code for the HIP estimator when applied to the HLL sketch is provided in Figure~\ref{HIPHLL:alg}.

\begin{algorithm}
\caption{HIP on HyperLogLog sketches:
$k$-partition, base-2,  5-bit counters \label{HIPHLL:alg}}
\SetKwFunction{bucket}{$\ppart$}
\SetKwFunction{rank}{r}
\KwIn{Stream of elements from domain $V$ \\ random hash functions for $v\in V$:
  $\bucket{v} \sim U[1,\ldots,k]$,  $\rank(v) \sim U[0,1]$}
\SetKwArray{mm}{M}
\KwOut{Array $\mm$ of $k$ 5-bit registers; HIP estimate $c$}
\DontPrintSemicolon
\tcp{ Initialize the sketch}
\For{$i=1,\ldots,k$}{$\mm{i}\gets 0$}
$c \gets 0$\tcp*{approximate count}
\tcp{Process stream element with value $v$}
 $h(v) \gets \min\{31,\lceil -\log_2 \rank(v)\rceil\}$\tcp*{$h(v) \in [0,\ldots,31]$}
\If {$h(v) > \mm{\ppart(v)}$}{
 $c \gets c+ \left(\sum_{i=1}^k I_{\mm{i} < 31}
   2^{-\mm{i}}\right)^{-1}$ \\ \;
 $\mm{\bucket{v}}\gets h(v)$}
\end{algorithm}

\ignore{
\begin{figure}[htbp]
\framebox{
\begin{minipage}{3in}
\begin{algorithmic}[1]
\Require Random uniform hash functions: $\ppart(v):[k]$, $r(v)$: first 32 bits
of $U[0,1]$.
\Statex
\State {\bf Initialization:}
\For {$i=1,\ldots,k$}  $M[i]\gets 0$ \Comment{$M[i]$ are 5-bit
  registers}
\EndFor
\State $c\gets 0$  \Comment{$c$ is an approximate counter}
\Statex
\State {\bf Processing stream element $v$:}
 \State $h(v) \gets \min\{31,\lceil -\log_2 r(v)\rceil\}$
\If {$h(v)>M[\ppart(v)]$}
\State $c \gets c+ \left(\sum_{i=1}^k I_{\M{i} < 31} 2^{-\M{i}}\right)^{-1}$
\State $M{\ppart{v}}\gets h(v)$
\EndIf
\end{algorithmic}
\end{minipage}
}
\caption{Pseudo code for HIP on the HyperLogLog MinHash sketches:
$k$-partition, base-2, each register uses 5 bits.  To apply HIP we
  maintain an additional  register $c$. \label{HIPHLL:alg}}
\end{figure}
}

Figure \ref{HLLHIP:fig}  shows results for the performance of the HIP
and HLL estimators.  Noting again that each simulation can be
performed on any stream of distinct elements (multiple occurrences do
not update the sketch or the estimate).  We implemented 
HyperLogLog using the pseudocode provided in \cite{hyperloglog:2007}.
We show both the raw estimate and the improved bias corrected
estimate as presented.
The Figure also shows the back-of-the-envelope approximate bound we calculated for
the CV of HIP, $\sqrt{\frac{b+1}{4(k-1)}}$, and we can see that
it approximately matches simulation results.

A more recent and more complicated implementation of HLL
\cite{hyperloglogpractice:EDBT2013}
obtains improved performance.  The improvement amounts to  smoothing out the ``bump'' due to the
somewhat ad-hoc bias reducing component, but the asymptotic behavior is the same as
the original hyperLogLog.  We can see that HIP obtains an asymptotic
improvement over HLL and also has a smooth behavior.  Moreover, HIP is
unbiased  (unless all counters are saturated) and elegant, and does not
require corrections and patches as with \cite{hyperloglog:2007,hyperloglogpractice:EDBT2013}.

 We quantify the improvement more precisely in terms of the number $k$ of registers:
  The NRMSE of HLL is $\approx 1.08/\sqrt{k}$ versus $\approx \sqrt{3/(4k)}  \approx 0.866/\sqrt{k}$ of HIP.  This means that an HLL estimator requires
  $\approx 0.56k$ more registers for the same square error as a HIP estimator.
As discussed above, HIP requires an   
 additional register $c$, but its
  benefit, in terms of accuracy outweighs the overhead.

Some encoding optimizations that were proposed for HLL
\cite{hyperloglogpractice:EDBT2013} and elsewhere \cite{KNW:PODS2010}
can also be integrated with HIP.   In particular, the content of the $k$  registers is highly correlated can be
  represented compactly by storing only one value and offsets 
for others (the expected
  size of each offset is constant).  
Recall that the ``exponent'' component of the approximate
  count $c$ can also be represented as an offset (see discussion
  above).

  We note that HIP permits us to work with a different base, and get further
  improvements with respect to HyperLogLog.   Consider using base
  $b=2^{1/i}$ for $i\geq 1$.   With smaller base, we need larger
  counters but we also have a smaller variance.
We need about $\log_2\log_b n$ bits per register, for counting
  up to cardinality (number of distinct elements) $n/16$ (since we want to
  have the counters large enough so that at most a fraction of
  them get saturated).  Since
$\log_2\log_b n = \log_2(\log_2 n /\log_2 b) = \log_2\log_2 n -
\log_2\log_2 b \approx \log_2\log_2 n +\log_2 i$, it means we need
about $\log_2 i$ additional bits per register relative to base-2.
The CV is   $\approx \sqrt{\frac{b+1}{4(k-1)}}$.  
So with $i=1$ (base
$b=2$) we had CV $\approx 0.866/\sqrt{k}$, with $i=2$ ($b=\sqrt{2}$),
we need 1 additional  bit per register but 
the CV is $\approx 0.777/\sqrt{k}$, meaning that we need $20\%$ fewer registers for the same error as when using base-2.  The advantage of  base-$\sqrt{2}$
  kicks in when $n$ exceeds about $3\times 10^8$.
If the counting algorithm also retains a sample of distinct elements
(such as with reservoir sampling) and thus IDs of sampled elements are
  retained, representation size is dominated by $k\log n$ in which
  case we might as well use full ranks for the approximate distinct count.

 Our evaluation aims at practice, and we note the relation to
theoretical results.  First, our analysis applies to
random hash functions, and is justified by simulation results with
standard generators.   A classic lower bound of 
Alon et al. \cite{ams99} on the sketch size
is logarithmic in the cardinality
(and there is a matching upper bound by Kane et al. \cite{KNW:PODS2010}).
The HLL sketch, however, has a much smaller, double logarithmic, size.
The reason is that the lower bound ``includes'' the
encoding of the hash function as part of the sketch, a requirement
which is not justified when many distinct counters use the same hash function.

    Lastly, we comment on the mergeability of our extended MinHash
  sketches.  Mergeability means that we can obtain a sketch of the
  union of (overlapping) data sets from the sketches of the sets.
  This property is important when parallelizing or distributing the
  computation.
The MinHash component of the extended sketch are mergeable, but to
correctly merge the counts, we need to estimate the overlap
between the sets.  This can be done using the similarity estimation
  hat of MinHash sketches, and we plan to address it in future work.

\section{Approximate counting} \label{approxcount:sec}
An approximate  counter is applied to a stream
of positive integers $\{w_i\}$ and represents $n=\sum_i w_i$ approximately.
Whereas an exact representation takes $\lceil \log_2 n\rceil$ bits, 
an approximate counter that uses only $O(\log\log n)$ bits was
proposed by Morris \cite{Morris77} and analysed and extended 
by Flajolet \cite{Flajolet:BIT85}. 
This {\em Morris counter} is an integer $x\geq 0$ and the
estimate is $\hat{n}=b^x-1$, where the (fixed) base 
$b>1$ controls a tradeoff between approximation quality and
representation size.

 Approximate counters were originally presented only for increments of
 $1$.  We provide procedures here for
efficient weighted updates and for merges of two counters.
An increase of $Y$ to a counter $x$  is performed as follows:  Let $i \gets \lfloor \log_b
(Y/b^x +1) \rfloor$ be the maximum such that increasing the counter by
$i$ would increase the estimate by at most $Y$.  We then compute the
leftover $\Delta\gets Y-b^x(b^i-1)$ and update $x\gets x+i$.  Lastly, we
increase $x$ by $1$ with probability $\Delta/(b^x(b-1))$ (this is an
inverse probability estimate of $\Delta$). Merge of two Morris
counters $x_1,x_2$ is handled the same as incrementing $x_1$ with $b^{x_2}-1$.
The estimator $\hat{n}$ is clearly unbiased (by induction on
updates).

 It is easy to show that the
variance is dominated by the analysis in
\cite{Morris77,Flajolet:BIT85} for increments, since it can
only improve when two consecutive updates are combined to a
single update with the sum of their values. Intuitively, only the ``leftover'' part of the updates
contributes to the variance at all.  

For the application of HIP approximate distinct counters, the
approximate counter is used to accumulate the sum of the HIP estimates.
In this particular case, the magnitudes of the updates are increasing and
typically are about $1/k$ of the total.  Therefore, with the
choice of $b\leq 1+1/k$, the variance is significantly lower than for
the unit increments as analysed in \cite{Morris77,Flajolet:BIT85}.
 The number of bits needed for counter representation is $\log_2\log_b n \approx \log_2\log_2 n + \log_2 (1/(b-1))$ and the CV is about $(b-1)$.  When using $b=1+1/2^j$ we obtain that with $j$ additional bits in the representation we can obtain relative error of $1/2^j$.  

\section{Cardinality estimator from ADS size} \label{adssizeest:section}
We derive the {\em size}
estimator which is the unique unbiased cardinality estimator that is only based
on the size (number of entries) of the ADS.  Specifically, to estimate
the cardinality of $N_d(v)$, we apply the estimator to the number of entries in
$\ADS(v)$ with distances at most $d$.  In a streaming context, the
size estimator is applied to the
number of updates (which resulted in modifying) the MinHash sketch.
The size estimator is weaker than the HIP estimator, but uses less
information. It is of interest in a setting where one can
observe the approximate counter as a black box,  only observing the
number of modifications.

 The estimator we derive below 
is applied to the {\em number of entries} in a bottom-$k$ ADS that are within
distance at most $d$ from $v$.  The estimator assumes that the ADS is
computed with respect to ``unique'' distances.  That is, we apply some
symmetry breaking and ADS may include multiple nodes of same distance.

\begin{lemma}
The unique unbiased estimator $E_s$ of $|N_d(v)|$ based solely on ADS size $s=| N_d(v)  
  \cap \ADS(v) |$ is
$$E_s = \left\{ \begin{array}{ll} s\leq k \quad \text{:} \quad & s \\
                                                 \text{otherwise
                                                 }\quad \text{:} \quad
                                                 & k (1+\frac{1}{k})^{s-k+1}-1\ .
     \end{array}
                                             \right.$$
\end{lemma}
\begin{proof}
Let $C_{i,\ell}$ be the probability that exactly $i$ nodes are sampled
from the first $\ell$.  
For $\ell\geq k$ and $i< k$ or for $\ell\leq k$ and $i<\ell$,  $C_{i,\ell}=0$.
If $\ell\leq k$, then $C_{\ell,\ell}=1$.
  We have the relations
\begin{align*}
& \ell>k \ \text{:}\, & C_{\ell,\ell} =&  (k/\ell) C_{\ell-1,\ell-1}\\
& k<i<\ell\ \text{:} \, &  C_{i,\ell} =& (1-k/\ell) C_{i,\ell-1} + (k/\ell) C_{i-1,\ell-1} \\
& k<\ell\ \text{:}\,  & C_{k,\ell} =&  (1-k/\ell) C_{k,\ell-1}
\end{align*}
$(k/\ell)$ is the probability that the $\ell$th node is one of the first $k$
in the random permutation induced on the $\ell$ nodes closest to $v$.

  If $s<  k$, which is only possible if $n_r(v)=s$,  we have $E_s=s$.  
  If $s=  k$, to be unbiased for the case where $n_r(v)=k$ and this is
  the only possible count,  we have $E_k=k$.  
  Otherwise, for $s>k$, we have that any estimator that is unbiased on
  neighborhoods of size $s$  must satisfy
$s = \sum_{i=k}^s E_i C_{i,s}$, which we rearrange to obtain
\begin{equation} \label{ubforsize:eq}
E_s = \frac{s- \sum_{i=k}^{s-1} E_i C_{i,s}}{C_{s,s}}\ .
\end{equation}
We iteratively apply \eqref{ubforsize:eq} to uniquely determine $E_s$
for $s\geq k+1$.
To determine $E_{k+1}$,  we consider the
two possible ADS counts are $k$ and $k+1$ with respective probabilities
$C_{k,k+1}=1/(k+1)$ and $C_{k+1,k+1}=k/(k+1)$.  From \eqref{ubforsize:eq}
$k+1 = E_k C_{k,k+1} + E_{k+1} C_{k+1,k+1} = k/(k+1) + E_{k+1}
k/(k+1)$.  We obtain $E_{k+1}=(k+1)^2/k-1$.
It can be verified that the general solution satisfies 
$E_s = k (1+\frac{1}{k})^{s-k+1}-1\ .$
\end{proof}
  This estimator is also applicable with $k=1$, in which case it is
  simply $2^s$.   




\section{Extension: Non-uniform node weights} \label{nonuniformweights:sec}
For simplicity, our presentation assumed uniform node
weights.  We briefly discuss here the extension to arbitrary
nonnegative node weights.
 The closeness centrality definition 
 \eqref{closenesscdef} incorporates 
 node weights $\beta(j)$. We can also consider computing neighborhood weights 
 $$n_d(v) = \sum_{j | d_{vj} \leq d} \beta(j)\ .$$

To obtain the same CV guarantees
as estimators for uniform weights,  we need to
compute the ADSs with respect to the weights $\beta(i)$.
To do that, we draw the rank $r(i)$ for node $i$  using the exponential distribution with parameter 
$\beta(i)$ \cite{ECohen6f,bottomk07:ds}.  This is the same as drawing 
uniform ranks $r'(i)$ and using ranks $r(i)=-\ln(1-r'(i))/\beta(i)$. 
With these ranks, nodes with higher $\beta$ values have higher inclusion probabilities. 
The same ADS definitions and algorithms apply, simply using 
the modified ranks.  Note however, that ADSs can have larger expected sizes (the $\beta$ weights can be viewed 
as emulating multiple copies of a node).  

We first discuss MinHash cardinality estimators. 
The $k$-mins basic 
estimator (with exponentially distributed ranks) applies with same CV 
of $1/\sqrt{k-2}$
to weighted $k$-mins MinHash sketches \cite{ECohen6f}.  
An estimator for weighted bottom-$k$ MinHash sketches was 
given in 
\cite{bottomk:VLDB2008} for weighted sampling without 
 replacement and general order samples.  
An alternative with bottom-$k$ is to use $r(i)= r'(i)/\beta(i)$, which 
corresponds 
to Sequential Poisson (Priority) sampling 
\cite{Ohlsson_SPS:1998,DLT:jacm07}.  

Our HIP estimators naturally extend, and remain unbiased,  with any
weight-based rankings
$r(i)=f(i,\beta(i),r'(i))$:  When the sketch is modified, we compute
the probability the probability, which depends now on $\beta(i)$, that
$r(i)$ is below the threshold value.
If $r(i)$ are exponentially distributed, the CV of estimating
neighborhood weights and centralities is at most $1/\sqrt{2(k-1)}$.



\subsection*{Conclusion}
  
\ignore{
  Performing computation on massive graphs is a 
  practically important problem.  To apply to graphs with billions of edges,
computation should scale near linearly with graph size.  

The distance 
distribution and its derivatives had been extensively used to study
social and other graphs.
  Fast algorithms for estimating the distance distribution had existed
  for two decades and several implementations are available.  
We present a unified view of existing approaches and
obtain improved and asymptotically optimal estimates  for the distance distribution,
using the same computation.  
}
ADSs, introduced two decades ago, are emerging as a powerful
tool for scalable analysis of massive graphs and data streams.  
We introduce HIP estimators, which apply to an extensive class of
natural statistics, are simple to apply with all sketch flavors, 
 and significantly improve over state of the art.
For neighborhood cardinalities and closeness centralities, HIP
estimators 
have at most half the variance of previous estimators.\notinproc{  HIP can be
integrated and enhance accuracy in existing ADS implementations \cite{hyperANF:www2011,DSGZ:nips2013}.}
For approximate distinct counting on data streams,
HIP estimators outperform state of the art practical estimators.
In follow-up work, we applied 
HIP for ADS-based estimation of closeness similarity of two nodes
\cite{CDFGGW:COSN2013} and timed-influence of a set of seed nodes
\cite{timedinfluence:2014}.  The application
\cite{timedinfluence:2014} used ADS structures defined with respect to
multiple graphs.
In independent later work, Ting \cite{Ting:KDD2014} also proposed HIP in the
context of distinct counting.


Lastly, we obtained interesting insights on MinHash sketch-based
cardinality estimation by applying classic results from the theory of
point estimation.  
We expect this powerful theory, perhaps with some adaptations to
discrete settings, to provide insights on other sketch structures.

\notinproc{
\subsection*{Acknowledgement}
The author would like to thank Seba Vigna for helpful pointers.
}
{\small 
\bibliographystyle{plain}
\bibliography{cycle}

\begin{thebibliography}{10}

\bibitem{ams99}
N.~Alon, Y.~Matias, and M.~Szegedy.
\newblock The space complexity of approximating the frequency moments.
\newblock {\em J. Comput. System Sci.}, 58:137--147, 1999.

\bibitem{BackstromBRUV12:websci2012}
L.~Backstrom, P.~Boldi, M.~Rosa, J.~Ugander, and S.~Vigna.
\newblock Four degrees of separation.
\newblock In {\em WebSci}, pages 33--42, 2012.

\bibitem{BJKST:random02}
Z.~Bar-Yossef, T.~S. Jayram, R.~Kumar, D.~Sivakumar, and L.~Trevisan.
\newblock Counting distinct elements in a data stream.
\newblock In {\em RANDOM}. ACM, 2002.

\bibitem{Bavelas:HumanOrg1948}
A.~Bavelas.
\newblock A mathematical model for small group structures.
\newblock {\em Human Organization}, 7:16--30, 1948.

\bibitem{RaoBlackwell1947}
D.~Blackwell.
\newblock Conditional expectation and unbiased sequential estimation.
\newblock {\em Annals of Mathematical Statistics}, 18(1), 1947.

\bibitem{hyperANF:www2011}
P.~Boldi, M.~Rosa, and S.~Vigna.
\newblock {HyperANF}: Approximating the neighbourhood function of very large
  graphs on a budget.
\newblock In {\em WWW}, 2011.

\bibitem{BoldiVigna:IM2014}
P.~Boldi and S.~Vigna.
\newblock Axioms for centrality.
\newblock {\em Internet Mathematics}, 2014.

\bibitem{BrEaJo:1972}
K.~R.~W. Brewer, L.~J. Early, and S.~F. Joyce.
\newblock Selecting several samples from a single population.
\newblock {\em Australian Journal of Statistics}, 14(3):231--239, 1972.

\bibitem{BRODER:sequences97}
A.~Z. Broder.
\newblock On the resemblance and containment of documents.
\newblock In {\em Proceedings of the Compression and Complexity of Sequences},
  pages 21--29. IEEE, 1997.

\bibitem{Broder:CPM00}
A.~Z. Broder.
\newblock Identifying and filtering near-duplicate documents.
\newblock In {\em Proc.of the 11th Annual Symposium on Combinatorial Pattern
  Matching}, volume 1848 of {\em LNCS}, pages 1--10. Springer, 2000.

\bibitem{ECohen6f}
E.~Cohen.
\newblock Size-estimation framework with applications to transitive closure and
  reachability.
\newblock {\em J. Comput. System Sci.}, 55:441--453, 1997.

\bibitem{CDFGGW:COSN2013}
E.~Cohen, D.~Delling, F.~Fuchs, A.~Goldberg, M.~Goldszmidt, and R.~Werneck.
\newblock Scalable similarity estimation in social networks: Closeness, node
  labels, and random edge lengths.
\newblock In {\em COSN}. ACM, 2013.

\bibitem{classiccloseness:COSN2014}
E.~Cohen, D.~Delling, T.~Pajor, and R.~F. Werneck.
\newblock Computing classic closeness centrality, at scale.
\newblock In {\em COSN}. ACM, 2014.

\bibitem{timedinfluence:2014}
E.~Cohen, D.~Delling, T.~Pajor, and R.~F. Werneck.
\newblock Timed influence: Computation and maximization.
\newblock Technical Report cs.SI/1410.6976, arXiv, 2014.

\bibitem{CoKa:jcss07}
E.~Cohen and H.~Kaplan.
\newblock Spatially-decaying aggregation over a network: {M}odel and
  algorithms.
\newblock {\em J. Comput. System Sci.}, 73:265--288, 2007.
\newblock Full version of a SIGMOD 2004 paper.

\bibitem{bottomk07:ds}
E.~Cohen and H.~Kaplan.
\newblock Summarizing data using bottom-k sketches.
\newblock In {\em ACM PODC}, 2007.

\bibitem{bottomk:VLDB2008}
E.~Cohen and H.~Kaplan.
\newblock Tighter estimation using bottom-k sketches.
\newblock In {\em Proceedings of the 34th VLDB Conference}, 2008.

\bibitem{CoSt:pods03f}
E.~Cohen and M.~Strauss.
\newblock Maintaining time-decaying stream aggregates.
\newblock {\em J. Algorithms}, 59:19--36, 2006.

\bibitem{CoWaSu}
E.~Cohen, Y.-M. Wang, and G.~Suri.
\newblock When piecewise determinism is almost true.
\newblock In {\em {Proc. Pacific Rim International Symposium on Fault-Tolerant
  Systems}}, pages 66--71, December 1995.

\bibitem{CGLM:ICSR2011}
P.~Crescenzi, R.~Grossi, L.~Lanzi, and A.~Marino.
\newblock A comparison of three algorithms for approximating the distance
  distribution in real-world graphs.
\newblock In {\em TAPAS}, 2011.

\bibitem{Dangalchev:2006}
Ch. Dangalchev.
\newblock Residual closeness in networks.
\newblock {\em Phisica A}, 365, 2006.

\bibitem{DSGZ:nips2013}
N.~Du, L.~Song, M.~Gomez-Rodriguez, and H.~Zha.
\newblock Scalable influence estimation in continuous-time diffusion networks.
\newblock In {\em NIPS}. Curran Associates, Inc., 2013.

\bibitem{DLT:jacm07}
N.~Duffield, M.~Thorup, and C.~Lund.
\newblock Priority sampling for estimating arbitrary subset sums.
\newblock {\em J. Assoc. Comput. Mach.}, 54(6), 2007.

\bibitem{DurandF:ESA03}
M.~Durand and P.~Flajolet.
\newblock Loglog counting of large cardinalities (extended abstract).
\newblock In {\em ESA}, 2003.

\bibitem{EW_centrality:SODA2001}
D.~Eppstein and J.~Wang.
\newblock Fast approximation of centrality.
\newblock In {\em SODA}, pages 228--229, 2001.

\bibitem{Feller2}
W.~Feller.
\newblock {\em An introduction to probability theory and its applications},
  volume~2.
\newblock John Wiley \& Sons, New York, 1971.

\bibitem{Flajolet:BIT85}
P.~Flajolet.
\newblock Approximate counting: A detailed analysis.
\newblock {\em BIT}, 25, 1985.

\bibitem{hyperloglog:2007}
P.~Flajolet, E.~Fusy, O.~Gandouet, and F.~Meunier.
\newblock Hyperloglog: The analysis of a near-optimal cardinality estimation
  algorithm.
\newblock In {\em Analysis of Algorithms (AOFA)}, 2007.

\bibitem{FlajoletMartin85}
P.~Flajolet and G.~N. Martin.
\newblock Probabilistic counting algorithms for data base applications.
\newblock {\em J. Comput. System Sci.}, 31:182--209, 1985.

\bibitem{hyperloglogpractice:EDBT2013}
S.~Heule, M.~Nunkesser, and A.~Hall.
\newblock {HyperLogLog} in practice: Algorithmic engineering of a state of the
  art cardinality estimation algorithm.
\newblock In {\em EDBT}, 2013.

\bibitem{HT52}
D.~G. Horvitz and D.~J. Thompson.
\newblock A generalization of sampling without replacement from a finite
  universe.
\newblock {\em Journal of the American Statistical Association},
  47(260):663--685, 1952.

\bibitem{KNW:PODS2010}
D.~M. Kane, J.~Nelson, and D.~P. Woodruff.
\newblock An optimal algorithm for the distinct elements problem.
\newblock In {\em PODS}, 2010.

\bibitem{LehSche1950}
E.~L. Lehmann and H.~Scheff\'e.
\newblock Completeness, similar regions, and unbiased estimation.
\newblock {\em Sankhya}, 10(4), 1950.

\bibitem{LiChurchHastie:NIPS2008}
P.~Li, , K.~W. Church, and T.~Hastie.
\newblock One sketch for all: Theory and application of conditional random
  sampling.
\newblock In {\em NIPS}, 2008.

\bibitem{LOZ:NIPS2012}
P.~Li, A.~B. Owen, and C-H Zhang.
\newblock One permutation hashing.
\newblock In {\em NIPS}, 2012.

\bibitem{pregel:sigmod2010}
M.~H. Malewicz, G.and~Austern, A.J.C Bik, J.~C. Dehnert, I.~Horn, N.~Leiser,
  and G.~Czajkowski.
\newblock Pregel: a system for large-scale graph processing.
\newblock In {\em SIGMOD}. ACM, 2010.

\bibitem{Morris77}
R.~Morris.
\newblock Counting large numbers of events in small registers.
\newblock {\em Comm. ACM}, 21, 1977.

\bibitem{Naiad:sosp2013}
D.~G. Murray, F.~McSherry, R.~Isaacs, M.~Isard, P.~Barham, and M.~Abadi.
\newblock Naiad: a timely dataflow system.
\newblock In {\em SOSP}, 2013.

\bibitem{Ohlsson_SPS:1998}
E.~Ohlsson.
\newblock Sequential poisson sampling.
\newblock {\em J. Official Statistics}, 14(2):149--162, 1998.

\bibitem{Opsahl:2010}
T.~Opsahl, F.~Agneessens, and J.~Skvoretz.
\newblock Node centrality in weighted networks: Generalizing degree and
  shortest paths.
\newblock {\em Social Networks}, 32, 2010.
\newblock \url{http://toreopsahl.com/2010/03/20/}.

\bibitem{PGF_ANF:KDD2002}
C.~R. Palmer, P.~B. Gibbons, and C.~Faloutsos.
\newblock {ANF:} {A} fast and scalable tool for data mining in massive graphs.
\newblock In {\em KDD}, 2002.

\bibitem{Rosen1997a}
B.~Ros{\'e}n.
\newblock Asymptotic theory for order sampling.
\newblock {\em J. Statistical Planning and Inference}, 62(2):135--158, 1997.

\bibitem{Rosenblatt:stats1956}
M.~Rosenblatt.
\newblock Remarks on some nonparametric estimates of a density function.
\newblock {\em The Annals of Mathematical Statistics}, 27(3):832, 1956.

\bibitem{Ting:KDD2014}
D.~Ting.
\newblock Streamed approximate counting of distinct elements: Beating optimal
  batch methods.
\newblock In {\em KDD}. ACM, 2014.

\end{thebibliography}
}

 \onlyinproc{
\begin{IEEEbiography}{Edith Cohen}
 Edith Cohen is a visiting full professor at Tel
Aviv University.  She received a Ph.D in Computer Science from
Stanford University in 1991. From 1991 to 2012 she was at AT\&T
Labs (initially AT\&T Bell Laboratories) and from 2012 to  2014 she
was a Principal Researcher at Microsoft Research (Silicon Valley). She was a visiting professor
at UC Berkeley in 1997.  Her research interests include algorithms,
mining and analysis of massive data, optimization, and computer
networking.  She is a winner of the IEEE ComSoc 1997 Bennett prize,
and an author of over 20 patents and hundreds of publications. 
\end{IEEEbiography}

\end{document}}
\appendices

\section{Extension: ADS without tie breaking} \label{notieADS:sec}

 We provide here a modified ADS definition, and respective HIP
 probabilities, for when there is a smaller set of distinct distances.  The advantages of
the modified definition is a smaller ADS size:  The modified ADS (we
provide here the bottom-$k$ flavor) includes a subset of the entries
that would have been included under the original definition (with tie
breaking on distances), but at most $k$ entries (those with smallest ranks) from each distinct distance.

\ignore{
We can partition distances into discrete ranges: $A_i=[a_i,a_{i+1})$ where $a_1=0$ and $a_i$ are increasing.  The values of $a_i$ can be determined according
to the granularity in which we want to differentiate distances.

The $\{a_i\}$-approximate $\ADS$ is follows.
  A node $u$ with $d_{uv}\in A_i$ is included in $\ADS(v)$ if $r(u)$ is smaller than the $k$th lowest rank amongst nodes within distance smaller than $a_{i+1}$.
An approximate ADS includes a strict subset of the nodes included in an exact ADS.
}
 Formally, a node $u$ is included in (modified) $\ADS(v)$ if
 $r(u)$ is smaller than the $k$th lowest rank amongst nodes within
 distance {\em at most} $d_{vu}$ from $v$.

  We can assign HIP inclusion probabilities for the modified ADS as follows.

  For each $v,u$, we compute the probability of $u$, conditioned on fixed ranks of all other nodes excluding $u$, of $u$ having
one of the $k-1$ smallest ranks amongst nodes with distance 
in $d_{vu}$ from $v$.  
We compute this probability only for nodes that satisfy this condition of having one of these $k-1$ smallest ranks.
The threshold probability is that $k$th smallest rank.
  Note that a node $u\in \ADS(v)$ that has the 
$k$th smallest rank in $N_{d_{vu}}(u)$ is not considered ``sampled.''  

The modified HIP probabilities 
can be applied to the same queries.
The HIP probabilities of an entry in the modified ADS are at most
the values in the full with tie-breaking ADS.  Therefore, adjusted
weights and variances are higher.  
The CV is at most $1/\sqrt{k-2}$, this is because when all distances are the same (say edges have $0$ lengths and the ADS is a reachability sketch), the modified ADS is a bottom-$k$ 
MinHash sketch of the reachability set.
\ignore{
  When the partition is such that there is a single node in each $A_i$, then all nodes are ``sampled'' and the probabilities are the HIP probabilities.
When all nodes are in a single $A_i$, we obtain a bottom-$k$ sketch of the set
and only $k-1$ nodes obtain adjusted weights.

 Neighborhood cardinality estimates for $a_i$-neighborhoods obtain a CV that
is between $1/\sqrt{k-2}$ (basic estimates), when all nodes are in a single $A_j$ ($j<i$) in which case the ADS size is $k$, and $1/\sqrt{2k-2}$ when nodes are in distinct ``rings.''

 Computation may require propagating values that are not included in the node's ADS.  If this is not done, ADS of other nodes can have ``errors'' of $a_{i+1}-a_i$.
}

\section{More on ADS computation}
We discuss some additional  aspects of ADS computation.

\ignore{
A complete representation of ADS$(i)$, containing all IDs of included nodes 
and their distances from $i$ requires 
$O(\log n)$ bits (node ID, distance) for each 
included node, which sums up to expected total of $O(k\log^2 n)$ bits 
per node. 
Here, we assume that the rank value is obtained by applying a hash 
function to the node ID. 
}

\subsection{Limited ADS computation} \label{implicit:sec}
When memory is constrained, we can benefit when 
not maintaining the full ADS in 
 ``active'' memory, but only maintaining  the {\em threshold}
information required to proceed with the computation. 
We refer to this as a {\em limited} ADS computation. 

With {\sc PrunedDijkstra}, ranks are processed in increasing order. In the 
iteration from node $i$, when visiting a node, we need to have access 
to all rank-distance pairs in the ADS constructed so far at the 
visited node.  With {\sc DP}, processing is in increasing distance.
To determine if a proposed entry indeed contributes to the ADS, we
only need to maintain  the MinHash sketch of ranks presented so far,
which has size $k$.
The ANF \cite{PGF_ANF:KDD2002,CGLM:ICSR2011} and hyperANF 
\cite{hyperANF:www2011} algorithms are essentially limited {\sc DP}
computation with base-2 ranks.  These algorithm maintain in iteration
$i$ for each node $v$ the MinHash sketch of $N_i(v)$ (all nodes with
distance at most $i$ from $v$). 
Streaming approximate distinct count estimators 
\cite{FlajoletMartin85,DurandF:ESA03} were applied to the base-2 
MinHash sketch of the current $N_i(v)$ to estimate its cardinality after 
each {\sc DP}  iteration.  The results from different nodes were  aggregated after each iteration 
to produce an estimate of the total number of pairs within each 
distance. More accurate estimates (as demonstrated in Section \ref{distinctcount:sec}) can be obtained using
the same implementations by applying our HIP estimators instead.

\subsection{Cost of relaxations}
The expected  total number of  relaxations is 
$O(mk\log n)$ but the
expected number of relaxations that actually
result in an update is $O(nk\log n)$.  This distinction is important
because relaxations which result in an update are more costly.

We first consider relaxations with {\sc DP}.
With $k$-mins and $k$-partition ADS, we can retain with each update the
index (out of $k$) which was modified since the last update.  If we do
so, then the cost of relaxing an edge is
$O(1)$, since we only need to look at the rank value in the modified
index.
If the index is not retained, we can perform a coordinate wise minimum
of the $k$ entries, in time $O(k)$.  The better choice depends on the hardware.

With bottom-$k$ ADS,  
we maintain the current bottom-$k$ ranks in active memory.
When relaxing an edge we compare the newly inserted
rank value in the sink ADS. 
If the entries are maintained in a 
 max-heap, the maximum entry is compared with the new value.
If the new value is smaller, it is inserted into the heap and the max
entry is removed.
 The cost is $O(1)$ if the ADS is not updated (node is not inserted) but
$O(\log k)$ otherwise (the max node is removed from the heap and the
new node is inserted).  
Alternatively, we can maintain the $k$ values in a sorted list and
each update takes $O(k)$ time.

Relaxations with {\sc PrunedDijkstra} are less efficient than with
{\sc DP} as we need to search
for the minimum rank in a distance range which increases update times
by a factor of $\log k$.



\subsection{Removing unique distances assumption}
 The strict ADS is defined with respect to
  unique distances.   We can apply any symmetry breaking between nodes
  of equal distance, but to
  maintain efficiency, in particular for {\sc DP} computation, we specify a
  particular one (this is all for analysis purposes).  The symmetry
  breaking is defined according
to the scan order of incoming edges to $v$ in the representation of
the graph.    Amongst two nodes $u$ and $w$ so that
$x=d_{uv}=d_{wv}$, we consider all paths of length $x$ to $v$
originating from $u$ (or $w$) and associate with $u$ the
least ordered incoming edge to $v$.  The closer one of  $u$ and $w$ is
the defined as the one with the earlier edge. 
 If both have the same earliest edge $(y,v)$, we
consider the same order with respect to the common previous node $y$ on
the path, and so on.  If {\sc DP} performs relaxations according to this
order, then we maintain the property that inserted nodes are in the
final ADS (with respect to the ``unique'' order).  

\subsection{Parallelizing  {\sc PrunedDijkstra}}
As stated,  the algorithm
performs $n$ sequential Dijkstra computations.  The dependences can be improved.  Consider $k=1$:
We partition the nodes to two sets according to rank.  We then perform
the computation from the set of lower rank nodes collapsed together
to a single node.  This will provide us ADS entries and their distances
for the closest node in that batch.  After we do this, we can proceed
for the second batch without completely resolving the first set of nodes. 
 Recursing, this gives us logarithmic depth. Further details are in \cite{ECohen6f}.

\subsection{{\sc PrunedDijkstra} base-$b$}  When we work with bottom-$k$
sketches and base-$b$ ranks, which are not uniquely assigned for nodes, we have to ensure that the ADS is not 
updated twice in the same iteration.  This can be done by 
marking each node after the first visit in each iteration and stopping the 
search on subsequent visits.  Note that the threshold information
at each node can include multiple occurrences of the same base-$b$
rank value (each corresponding to a distinct node). 
But because iteration order corresponds to the order on the full rank,
the entries correspond to the entries of the full-rank ADS. 
To obtain the explicit ADS (with node IDs),
we can export each new entry $i\in \ADS(v)$ and the distance $d_{iv}$ to 
a slower medium. After the computation we can aggregate all entries of 
$\ADS(v)$ for each $v$. 

 With {\sc DP}, base-$b$, and bottom-$k$ sketches, we must treat rank values in 
the same ADS as unique.  This is needed to avoid having the same 
node contribute to multiple entries in an ADS of another node.

\end{document}